\documentclass[12pt]{article}

\usepackage{amsmath,amssymb}
\usepackage{multirow}
\usepackage[pdftex]{graphicx}
\usepackage{natbib}
\usepackage{mathrsfs}
\usepackage{bm, xcolor}
\usepackage{mathrsfs}
\usepackage{verbatim}
\usepackage{longtable}
\usepackage{threeparttable}
\usepackage{color}
\usepackage{amsmath}
\usepackage{url}

\usepackage{enumerate}
\usepackage{amssymb}
\usepackage{amsbsy}
\usepackage{amsmath}
\usepackage{amsthm}
\usepackage{amsfonts}
\usepackage{latexsym}
\usepackage{xifthen} 
\usepackage{color}
\usepackage{manfnt}
\usepackage{ifthen}
\usepackage{bm}
\usepackage{caption}
\usepackage{subcaption}

\usepackage{algorithm}
\usepackage{algorithmic}

\usepackage{float}

\usepackage{cleveref}
\crefformat{equation}{(#2#1#3)}
\crefrangeformat{equation}{(#3#1#4) to~(#5#2#6)}
\crefname{equation}{}{}
\Crefname{equation}{}{}

\def\calA{\mathcal{A}}

\def\calI{\mathcal{I}}

\def\calM{\mathcal{M}}
\newcommand{\mM}{\mathcal{M}}
\def\P{\mathbb{P}}

\newcommand{\wt}[1]{\widetilde{#1}} 

\newcommand{\mR}{\mathbb{R}}
\newcommand{\C}{\mathbb{C}} 

\newcommand{\E}{\mathbb{E}} 
\renewcommand{\P}{\mathbb{P}} 
\newcommand{\cov}{\mathop{\rm Cov}} 

\newcommand{\indic}[1]{\mbf{1}\left\{#1\right\}} 

\makeatletter
\def\singlespace{\def\baselinestretch{1}\@normalsize}

\newtheoremstyle{mythmstyle}
{8 pt} 
{3 pt} 
{} 
{} 
{\bfseries} 
{.} 
{.5em} 
{} 

\theoremstyle{plain}

\makeatletter
\def\thm@space@setup{%
	\thm@preskip=6pt plus 1pt minus 1pt
	\thm@postskip=\thm@preskip 
}
\makeatother

\newtheorem{theorem}{Theorem}[section]
\newtheorem{lemma}[theorem]{Lemma}

\newtheorem{remark}{Remark}

\newtheorem*{example*}{Example}

\newtheorem*{definition*}{Definition}

\newtheorem*{remark*}{Remark}
\crefname{definition}{\textbf{definition}}{definitions}
\Crefname{definition}{Definition}{Definitions}
\crefname{assumption}{\textbf{assumption}}{assumptions}
\Crefname{assumption}{Assumption}{Assumptions}

\def\singlespace{\def\baselinestretch{1}\@normalsize}

\makeatother

\def\newpage{\vfill\eject}
\def\wh{\widehat}

\def\wt{\widetilde}

\newdimen\biblioindent    \biblioindent=30pt

 at 7truept

\def\beqr{\begin{eqnarray}}
	\def\eeqr{\end{eqnarray}}
\def\beqrs{\begin{eqnarray*}}
	\def\eeqrs{\end{eqnarray*}}

\def\beq{\begin{equation}}
\def\eeq{\end{equation}}
\def\beqn{\begin{eqnarray}}
\def\eeqn{\end{eqnarray}}
\def\beqnn{\begin{eqnarray*}}
\def\eeqnn{\end{eqnarray*}}

\def\wh{\widehat}
\def\wt{\widetilde}

\newcommand{\defn}{\stackrel{\mbox{{\tiny def}}}{=}}

\newcommand{\trans}{^{\mbox{\tiny{T}}}}
\def\defby{\stackrel{\mbox{\textrm{\rm\tiny def}}}{=}}

\def\calA{\mathcal{A}}
\def\calI{\mathcal{I}}

\def\calM{\mathcal{M}}

\newcommand{\eps}{\mbox{$\varepsilon$}}

\def\indic{\mathbb I} 

\begin{document}
	\renewcommand{\baselinestretch}{1.3}

	\title {\bf  Robust Model Selection with Application in Single-Cell Multiomics Data }
	\author{   Zhanrui Cai  }
	\date{\empty}
	
	\maketitle
	\renewcommand{\baselinestretch}{1.5}
	\baselineskip=24pt
	\noindent{\bf Abstract:}

		Model selection is critical in the modern statistics and machine learning community. However,  most existing works do not apply to heavy-tailed data, which are commonly encountered in real applications, such as the single-cell multiomics data. In this paper, we propose a rank-sum based approach that outputs a confidence set containing the optimal model with guaranteed probability. Motivated by conformal inference, we developed a general method that is applicable without moment or tail assumptions on the data. We demonstrate the advantage of the proposed method through extensive simulation and a real application on the COVID-19 genomics dataset \citep{stephenson2021single}. To perform the inference on rank-sum statistics,  we derive a general Gaussian approximation theory for high dimensional two-sample U-statistics, which may be of independent interest to the statistics and machine learning community.

	\par \vspace{9pt} \noindent {\it Key words and phrases: } Heavy-tailed data; High-dimensional data; Model selection.

	\pagestyle{plain}
	
	\newpage
	
		\section{Introduction}
	
	With the advancement of techniques, modern statistics and machine learning communities developed numerous complex models to capture the underlying trends among datasets. Of particular interest are data with heavy-tailed distributions, which are ubiquitous in real applications. To achieve stability and robustness, a number of recent works have been proposed for different tasks in heavy-tailed data, including but not limited to mean estimation, high dimensional regression, low-rank matrix recovery, etc. See \cite{fan2014adaptive, fan2017estimation, lugosi2019mean, wang2020tuning, fan2021shrinkage} and references therein for examples. However, how to develop a general procedure for model selection in heavy-tailed data is still a  crucial yet challenging open problem in the literature, especially for high-dimensional datasets.
	
	Due to its simplicity and wide adaptivity, cross-validation is perhaps the most popular and practical tool for model selection. By training and testing the candidate models on different subsets of the data,  cross-validation aims to choose the model with the best prediction accuracy on testing datasets \citep{hastie2009elements}. Consider a finite set of candidate models $\calM $ with index $\{1, 2,\dots, M\}$, where each $m\in\calM$ can be a model or a tuning parameter value. The most common $K$-fold cross-validation begins with dividing the data into $K$ folds. Each candidate model is trained on the data excluding the $k$-th fold for $k=1,2,\dots, K$, and then validated by evaluating the empirical prediction risk of the fitted model on the $k$-th fold. The final cross-validated risk is then obtained by averaging the empirical prediction risks while rotating each fold as the hold-out set. The selected model is the one that has the smallest cross-validated prediction risk. 
	
	Despite its popularity, the literature has also shown that cross-validation may fail to select the best model consistently under many natural scenarios. For example, in the low dimensional linear regression, researchers have demonstrated that cross-validation is only consistent when the ratio of the training sample size and testing sample size goes to zero, which excludes the most popular $K$-fold or leave-one-out cross-validation \citep{shao1993linear, yang2007consistency}. In tuning parameter selection, researchers have found that cross-validation or AIC fails to select the penalization parameter that achieves the oracle property, and the BIC is preferred \citep{wang2007tuning, chen2008extended, lee2014model}. However, the information-based criterion usually requires clear definitions of the likelihood and the degree of freedom, which may be inapplicable for popular machine learning algorithms, such as neural networks. \cite{lei2020cross} proposed the cross-validation with confidence (CVC) that achieves the consistency of model selection with conventional sample splitting or V-fold cross-validation. While traditional cross-validated model selection can be viewed as a point estimator $\wh m_{CV}\in\calM$, CVC  provides confidence set ${\wh \calM}_{CVC}\subset\calM$ that includes the optimal model at a pre-specified probability level. The selected models in  ${\wh \calM}_{CVC}$ are all nearly optimal, and the preferred model can then be chosen based on specific rules or expert knowledge. For example, one may select the sparsest model in the high dimensional linear regression when sparsity is the desired feature.

	In practice, a caveat of CV or CVC is their limited applicability for heavy-tailed data, as the estimation of cross-validated risks and the corresponding asymptotic distribution requires non-trivial trail conditions on the data as well as the loss function considered. This paper is motivated by a COVID-19 genomics dataset \citep{stephenson2021single},  which entails heavy tail distributions for both the covariates and responses. The data contains measurements of RNA-seq and surface proteins of human blood immune cells. The RNA-seq is usually a high dimensional sparse vector, while the surface proteins are generally low dimensional and are of direct interest to scientists because they are functionally involved in cell signaling and cell-cell interactions \citep{davis2007intercellular}. While gene expression (RNA-seq) data has been extensively studied in the single-cell literature,  recently, researchers have been interested in building predictive models to predict surface proteins based on RNA-seq \citep{zhou2020surface, cai2022PNAS, cai2021asymptotic}, as well as general imputation methods for multiomics data \citep{du2022robust}. As illustrated in Figure \ref{fig: des_real_data}, both the gene expressions and proteins are heavy-tail distributed, and the moment or tail conditions can be easily violated.

	To achieve robust model selection, we propose a novel solution based on rank-sum, where we select the model confidence set with rank-sum comparison. The new method is inspired by recent development in conformal inference \citep{vovk2005algorithmic, lei2018distribution}, which guarantees valid type-I error control with minimal assumptions on the data and models. While previous literature on conformal prediction mainly focuses on providing valid prediction intervals, we show that it can also be applied in robust model selection for heavy-tailed data, as demonstrated both theoretically and numerically. To prove the consistency of the proposed method, we develop a general approximation theory for high dimensional two-sample U-statistics, which may be of independent interest to the statistics community.

	\begin{figure}
		\centering
		\begin{subfigure}[b]{0.475\textwidth}
			\centering
			\includegraphics[scale=0.4]{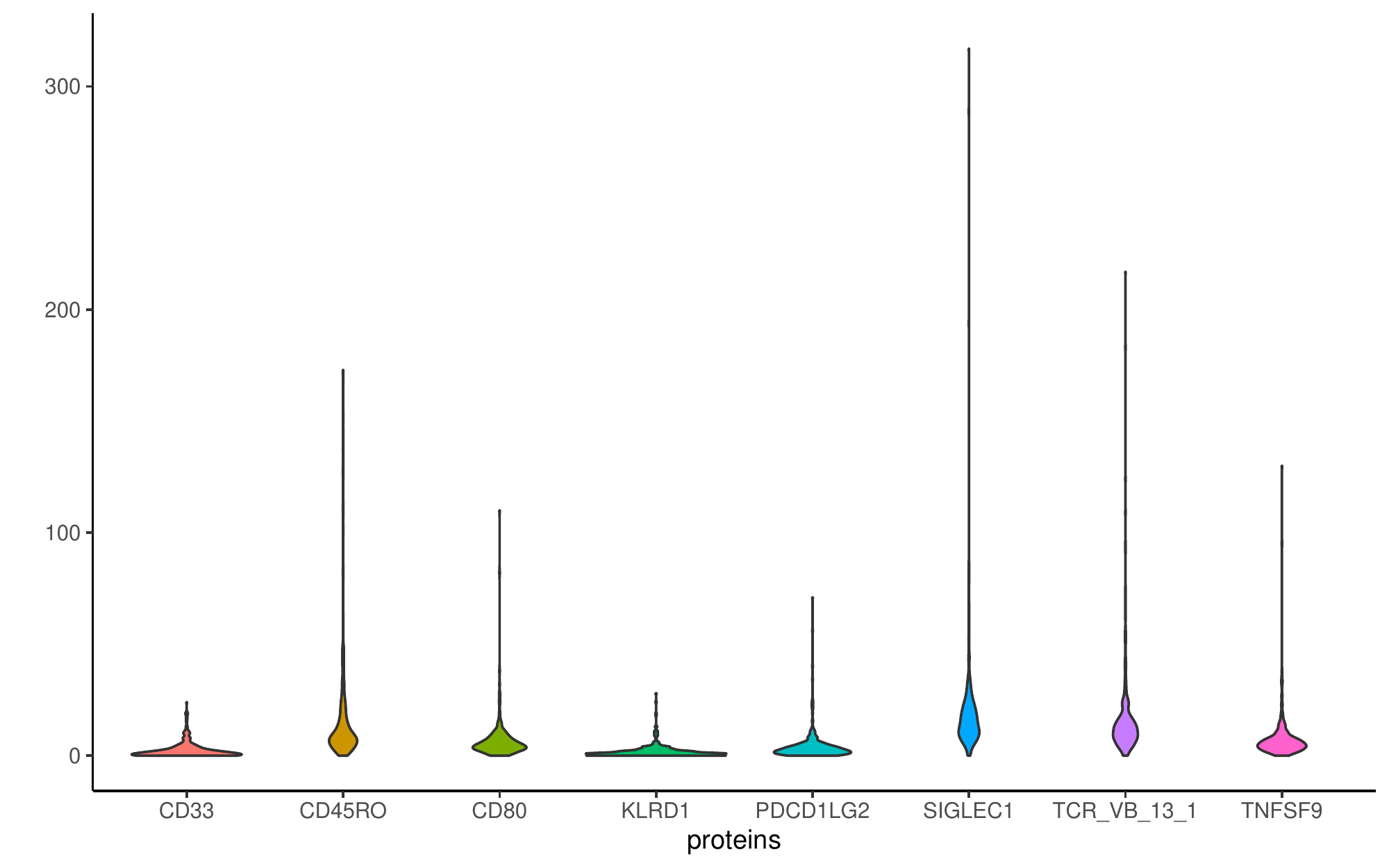}
		\end{subfigure}
		\hfill
		\centering
		\begin{subfigure}[b]{0.475\textwidth}
			\centering
			\includegraphics[scale=0.4]{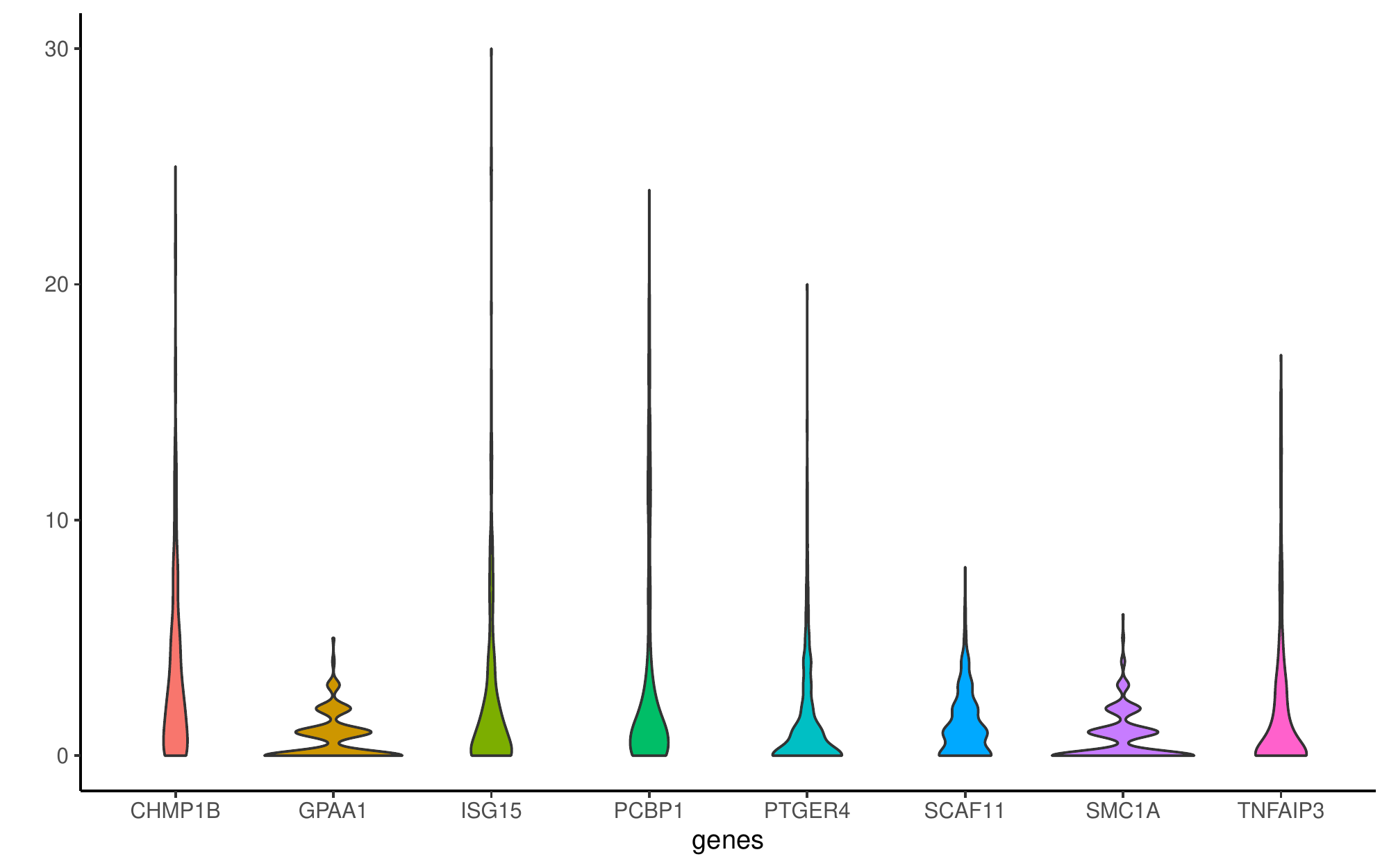}
		\end{subfigure}
		\caption{The violin plots of 8 randomly selected proteins (left) and genes (right) in the COVID-19 genomics dataset \citep{stephenson2021single}. }\label{fig: des_real_data}
	\end{figure}
	
	The approximation theory of the maximum (minimum) of the high dimensional vector is an important topic in the recent statistics literature. \cite{chernozhukov2013gaussian} established the conditions under which the distribution of the maximum of a high dimensional mean vector can be approximated by the maximum of the mean of a Gaussian random vector with the same covariance matrix. Related works include but are not limited to \cite{chernozhukov2015comparison, chernozhukov2017central}, etc. \cite{chen2018gaussian} extended the applicability of this approach to the high dimensional vector of one-sample U-statistics, where the kernel of the U-statistics is symmetric in its arguments. In our paper, we further extend this line of research in two directions: 1) we consider two-sample U-statistics due to the asymmetric nature of rank-sum comparison; 2) the two samples can possibly be weakly dependent on each other. The newly developed tool is readily applicable to other high-dimensional inference problems, such as two-sample testing, change point detection, the goodness of fit testing, etc.
	
	\section{Methodology}
	
	\subsection{Preliminary}
	
	We consider general model selection in supervised learning. Consider data $D = \{(X_i, Y_i), i=1,\dots, n\}$ independently drawn from a common distribution. Assume that $X_i\in\mR^{d}$, $Y_i\in\mR$, and
	\beqrs
	Y_i = f(X_i) + \eps_i,
	\eeqrs
	where $f:\mR^{d}\rightarrow \mR$ is an unknown function and $\E(\eps|X) = 0$. Let $\wh f_m$ be an estimate of $f$ based on a model $m\in\calM$, where $\mathcal M$ is a finite collection of candidate models. Cross-validation approximates the predictive risk for every  $\wh f_m$, then choose the model with the smallest estimated risk. Specifically, the predictive risk of $\wh f_m$ is 
	\beqrs
	R(\wh f_m) = \E\left[ l(\wh f_m(X), Y) \mid \wh f_m\right],
	\eeqrs
	where $(X, Y)$ is a new random draw from the original distribution of the data, and $l(\cdot, \cdot)$ is a loss function. Robust loss functions, such as the Huber loss \citep{sun2020adaptive}, should be implemented when the data has heavy-tail distributions.
	If $ R(\wh f_m)< R(\wh f_j)$ for all $j\neq m$,  $f_m$ is selected as the best model. To estimate $R(\wh f_m)$, the V-fold cross-validation begins with splitting the data index $\calI$ into $V$ subsets $\calI_1, \dots, \calI_V$.  For each $\calI_v$, $f_m$ is trained on the data that excludes  $\calI_v$, denoted as $\wh f_m^{(v)}$. Then the risk is evaluated based on the data in $\calI_v$. With a little abuse of notation, we write $R(\wh f_m^{(v)}) = |\calI_v|^{-1} \sum_{i\in\calI_v} l( \wh f_m^{(v)}(X_i), Y_i)$. Then the estimated risk $R(\wh f_m)$ can be obtained by 
	\beqrs
	R(\wh f_m) = \frac{1}{V}\sum_{v=1}^{V}R(\wh f_m^{(v)}).
	\eeqrs
	Existing cross-validation methods focus on comparing $R(\wh f_m)$ to decide the best candidate model, which is, essentially, comparing the prediction accuracy of all the candidate models with sample size $(1-1/V)n$. However, due to the integral nature of risk functions, cross-validation requires moment conditions on the loss function, which in turn implies certain moment conditions on the data distribution. For example, we require the existence of the second-order moments of $\eps$ when using a least square loss function. And the existence of the first-order moment of $\eps$ is necessary when using the Huber loss function. In other words, the risk function could become impractical in model evaluation when the data is heavy-tailed. This motivates us to consider weaker conditions for model selection in the heavy-tailed setting.
	
	\subsection{Robust cross-validation and conformal inference}
	
	The proposed approach is inspired by recent development in conformal inference and conformal prediction \citep{hu2020distribution, lei2018distribution}. We first consider a simple sample splitting scenario, where the index set $\calI = \{1, \dots, 2n\}$ is randomly split into two equal subsets $\calI_1$ and $\calI_2$, each with cardinality $n$. Let $D_{1} = \{(X_i, Y_i), i\in\calI_1\}$ and $D_{2} = \{(X_i, Y_i), i\in\calI_2\}$ be the two subsets of the data. We train all the candidate models based on $D_1$ and evaluate the fitted models on $D_2$. Specifically, let $\{\wh f_m: m\in\mM\}$ be the set of candidate models trained on $D_1$. 
	
	When two candidate models perform similarly to each other, we propose to treat them as {\it conformal} enough so that inference methods can not distinguish one from each other. In other words, the loss of the two candidate models tends to be exchangeable. Specifically, if $\wh f_m$ and $\wh f_j$ perform similarly on the testing set, then the conformal prediction theory implies that the ranking statistic
	\beqr\label{eq: conformal_U}
	\wt U_{mj} \defby \frac{1}{n}\sum_{k\in\calI_2}\indic\left\{l(\wh f_m(X_k), Y_k) <  l(\wh f_j(\wt X), \wt Y)\right\}
	\eeqr
	has an approximately $U(0,1)$ distribution. $(\wt X, \wt Y)$ is an additional random draw from the distribution. Thus in order to compare $\wh f_m$ and $\wh f_j$, it suffices to test whether $\wt U_{mj}$ follows a uniform distribution. If $\wt U_{mj}$ is significantly larger than 0.5 for all $j$, then $\wh f_m$ is a better choice compared to $\wh f_j$. In fact, $\wt U_{mj} $ can be interpreted as the conformal $p$-value of testing the null hypothesis that $\wh f_j$ performs similarly as $\wh f_m$, based on one single point $(\wt X, \wt Y)$.	Because the conformal $p$-value is a ranking statistic, it bypasses the integral nature of the risk function and provides a robust model comparison criterion, especially for heavy-tailed data.
	
	However, the power of the statistic (\ref{eq: conformal_U}) can be severely limited in practice because it is only based on one single evaluation of $\wh f_j$. To achieve asymptotically full power, we consider testing on multiple data points and aggregating the ranking statistics. Specifically, we define a generalized rank-sum statistic
	\beqr\label{eq: rank_sum_stat}
	\wh U_{mj} = \frac{1}{n^2}\sum_{k, l\in\calI_2} \indic\{ l(\wh f_m(X_k), Y_k) <  l(\wh f_j(X_l), Y_l)\},
	\eeqr
	where both $\wh f_m$ and $\wh f_j$ are evaluated on the entire testing dataset. Note that the two samples being compared in the indicator function in equation (\ref{eq: rank_sum_stat}) are highly dependent on each other, which marks a clear distinction from the classical rank-sum statistic and results in more complicated analysis for the asymptotic variance. 
	
	
	\begin{remark}
		Without loss of generality, we assume the value of the loss function follows a continuous distribution. These ties will occur with probability zero in the comparison (\ref{eq: rank_sum_stat}). When the loss functions have discrete values, we use a random tie-breaking when calculating the rank-sum statistics. The theory follows similarly with careful bookkeeping. 
	\end{remark}

	\subsection{Rank-sum based robust model selection}
	
	To achieve model selection, we propose to simultaneously test whether the aggregated statistics $\wh U_{mj}$ are significantly larger than 0.5 for all $j\in\calM$ and $j\neq m$. This is essentially a high dimensional testing problem for two sample U-statistics, especially when the number of candidate models $|\calM|$ is large. Define 
	\beqrs
	\mu_{m,j} = \P\left\{ l(\wh f_m(X), Y) <  l(\wh f_j(X'), Y')\mid \wh f_m, \wh f_j \right\} - 0.5,
	\eeqrs
	We consider the hypothesis testing problem:
	\beqr\label{RSR_SS}
	H_{0,m} : \min_{j\neq m} \mu_{m,j} \geq 0\quad \mbox{v.s.}\quad H_{1,m}:  \min_{j\neq m} \mu_{m,j} < 0.
	\eeqr
	The intuition behind the hypothesis test is that the best model should always have a smaller predictive loss compared to other models. Thus if $\wh f_m$ is an ideal model, $\mu_{m,j}$ should always be non-less than 0, and the estimated $p$-value $\wh p_{ss, m}$ should be larger than the type-I error rate $\alpha$. 
	After performing the hypothesis testing (\ref{RSR_SS}) for each $m\in\calM$, we output the selected confidence set under the sample splitting (ss) setting:
	\beqrs
	{\wh\calM}_{ss} = \{\wh f_m: m\in\calM, \wh p_{ss, m}\geq\alpha \}.
	\eeqrs
	One may switch the role of $H_0$ and $H_1$ in testing (\ref{RSR_SS}), and output the selected confidence set as $\{\wh f_m: m\in\calM, \wh p_{ss, m}\leq\alpha \}$.  This leads to a different interpretation of the selected sets and a balance of type-I and type-II errors. When missing a potentially good model in the confidence set is more severe, we treat it as the type-I error and formulate the hypothesis tests as in (\ref{RSR_SS}). 
	
	Because the proposed method is a {\bf R}ank-{\bf S}um based {\bf R}obust approach, we named it RSR for convenience. RSR shares a similar inference-based model selection idea as the method proposed in \cite{lei2020cross}. However, the method in \cite{lei2020cross} is based on the risk functions and is not applicable for heavy-tailed data. Besides, the test statistic (\ref{eq: conformal_U}) are two sample U-statistics, and no existing techniques can be applied to obtain the $p$-values of the high dimensional testing problem (\ref{RSR_SS}). Another robust choice of comparision is the pairwise criterion $\P\left\{ l(\wh f_m(X), Y) <  l(\wh f_j(X), Y)\mid \wh f_m, \wh f_j \right\}$. However, we show in the numerical studies that it does not perform as well as the RSR under the heavy-tailed setting.
	
	\subsection{The Gaussian multiplier bootstrap algorithm}
	
	The hypothesis testing problem (\ref{RSR_SS}) is challenging from many perspectives. Firstly, the number of candidate models $|\calM|$ can be of high dimension, and it is hard to approximate the minimum of a high dimensional vector. Secondly, because all models are trained and tested on the same observational data, the estimated losses are highly dependent on each other, with unknown dependency structures. Thirdly, the rank-sum statistic is a two-sample U-statistics with two dependent samples, and is more complicated to analyze compared to the sample mean. To solve these challenges, we develop a new approximation theory for the minimum of high dimensional two-sample U-statistics, with detailed theories shown in the following subsection. To estimate the $p$-values for (\ref{RSR_SS}) with minimum conditions on the data and the candidate models, we propose using the Gaussian multiplier bootstrap, as illustrated below.
	
	Define the elements in the rank-sum statistic
	\beqr\label{eq:xi_ss}
	\xi_{m,j}^{(k,l)} = \indic\left\{l(\wh f_m(X_k), Y_k) <  l(\wh f_j(X_l), Y_l)\right\}-0.5,
	\eeqr
	where $j\in\calM$ and $k,l\in\calI_2$.
	The goal is to study the distribution of
	\beqrs
	\min_{j\neq m}\frac{1}{n^2}\sum_{k\in\calI_2}\sum_{l\in\calI_2}	\xi_{m,j}^{(k,l)}. 
	\eeqrs

	\begin{itemize}
		\item[1.] Randomly split the index $\calI = \{1, \dots, 2n\}$ into two equal subsets $\calI_1$ and $\calI_2$. Train the candidate models on $i\in\calI_1$, and obtain $\{\wh f_j, j\in\calM\}$.
		\item[2.] For each $j\neq m$, estimate $\wh\mu_{m,j} = n^{-3/2}\sum_{k\in\calI_2}\sum_{l\in\calI_2}	\xi_{m,j}^{(k,l)} $, with $	\xi_{m,j}^{(k,l)}$ given by (\ref{eq:xi_ss}).  Let $T = \min_{j\neq m}\wh\mu_{m,j}$.
		\item[3.] For $b=1,\dots, B$,
		\begin{itemize}
			\item[(a)] Generate i.i.d. standard Gaussian random variables $e_k$.
			\item[(b)] Let 
			\beqrs
			T_b = \min_{j\neq m} \frac{1}{\sqrt n}\sum_{k\in\calI_2}\left[\frac{1}{n}\sum_{j\in\calI_2}\xi_{m,j}^{(k,l)} - \wh\mu_{m,j}  \right] e_k.
			\eeqrs
		\end{itemize}
		\item[4.] Obtain $p$-value: $\wh p_{ss, m} = B^{-1}\sum_{b=1}^{B}\indic(T_b<T)$.
	\end{itemize}
	
	The Gaussian multiplier bootstrap is essentially approximating the distribution of the high dimensional vector $(\mu_{m,j}, j\neq m)$. Thus the dimensionality and dependency are automatically taken into account. The algorithm can be easily extended to $V$-fold cross-validation and is summarized below.
	
	\begin{itemize}
		\item[1.] Randomly split the index $\calI = \{1, \dots, 2n\}$ into $V$ subsets $\calI_1, \dots, \calI_V$. For each $1\leq v\leq V$, train the candidate models on $i\notin\calI_v$, and obtain $\{\wh f_j^{(v)}, j\in\calM\}$.
		\item[2.] For $k,l\in\calI$ and $j\neq m$, calculate $$\xi_{m,j}^{(k,l)} = \indic\left\{l(\wh f^{(v_1)}_m(X_k), Y_k) <  l(\wh f^{(v_2)}_j(X_l), Y_l)\right\}-0.5.$$
		where $k\in\calI_{v_1}$ and $l\in\calI_{v_2}$.
		\item[3.] For each $j\neq m$, estimate $\wh\mu_{m,j} = (2n)^{-3/2}\sum_{k\in\calI}\sum_{j\in\calI}	\xi_{m,j}^{(k,l)} $. Let $T = \min_{j\neq m}\wh\mu_{m,j}$.
		\item[4.] For $b=1,\dots, B$,
		\begin{itemize}
			\item[(a)] Generate i.i.d. standard Gaussian random variables $e_k$.
			\item[(b)] Let 
			\beqrs
			T_b = \min_{j\neq m} \frac{1}{\sqrt{2n}}\sum_{k\in\calI_2}\left[\frac{1}{2n}\sum_{j\in\calI_2}\xi_{m,j}^{(k,l)} - \wh\mu_{m,j}  \right] e_k.
			\eeqrs
		\end{itemize}
		\item[5.] Obtain $p$-value: $\wh p_{V, m} = B^{-1}\sum_{b=1}^{B}\indic(T_b<T)$.
	\end{itemize}
	
	\subsection{Screening when the number of candidate models is large}\label{sec: screen}
	
	The computation cost can be expensive when the number of candidate models is huge. Motivated by the moment inequalities in \cite{chernozhukov2019inference}, we implement a screening procedure to eliminate the models that are clearly inferior to $f_m$. Define
	\beqr\label{screening}
	\wh J_m =  \Big\{j\neq m: \sqrt{n}\frac{\wh\mu_{m,j}}{\wh\sigma_{m,j}}\leq 2 c_{\alpha'} \Big \}, \quad\mbox{where}\quad c_{\alpha'} = \Phi^{-1}\Big(1-\frac{\alpha'}{(M-1)^{1+s}}\Big).
	\eeqr
	Here $s>0$ is a small constant tunning parameter and is mainly used for the convenience of theoretical analysis. In all the numerical studies, we set $s=0.01$. $\Phi$ is the distribution function of the standard normal distribution. Clearly, the models in $J_m^c$ have relatively large $\wh\mu_{m,j}$ and are inferior compared to $f_m$; thus there is no need to invoke the bootstrap comparison for the models eliminated out of $\wh J_m$.
	
	The screening set in (\ref{screening}) is similar to the one in formula (25) of \cite{chernozhukov2019inference}. However, the original method in \cite{chernozhukov2019inference} is only applicable to mean estimation, while our $\wh\mu_{m,j}$ are two sample U-statistics. We prove in the next section that the screening set can contain the optimal candidate model with high probability.
	
	Because the rank-sum statistic is comparing two dependent samples, $\wh\sigma_{m,j}$ is no longer equal to $1/6$ as in the classical settings. To estimate $\wh\sigma_{m,j}$, we define extra notation to avoid redundant parentheses in the equations. Specifically, let $A_k = l(\wh f_m(X_k), Y_k)$ and $B_{j, l} = l(\wh f_j(X_l), Y_l)$. We drop the notation on its dependence on $m$ for simplicity. Then 
	\beqrs
	\wh\mu_{m, j} = \frac{1}{n^2}\sum_{k\in\calI_2}\sum_{l\in\calI_2}	\indic\left\{A_k<B_{j, l} \right\}-0.5,
	\eeqrs
	with its projection given by $n^{-1}\sum_{i=1}^{n}\Big\{  F_{1}(B_{j,i})+ 1 -  F_{2,j}(A_i)\Big\}$, where $F_{1}$ and $F_{2, j}$ are the cumulative distribution function of $A_k$ and $B_{j, l}$, respectively. 
	Thus the variance and its empirical estimator are 
	\beqr\label{eq: sigma}
	&&\sigma_{m, j}^2 = \frac{1}{6n} - 2 \cov\{F_{1}(B_1),  F_{2,j}(A_1)\},\\\label{eq: sigma_hat}
	&&\wh\sigma_{m, j}^2 =  \frac{1}{6n} - \frac{2}{n}\Big( \sum_{i=1}^{n}\wh F_{1}(B_i) \wh F_{2, j}(A_i)  - \sum_{i=1}^{n} \wh F_{1}(B_i) \sum_{i=1}^{n}\wh F_{2,j}(A_i)     \Big),
	\eeqr	
	where $\wh F_{1}$ and $\wh F_{2,j}$ are the empirical estimator of $F_{1}$ and $F_{2,j}$.

	\subsection{Consistency analysis}
	
	In this paper, we prove the consistency of RSR under the sample splitting scenario. The analysis was conducted by conditioning on the fitted candidate models $\{\wh f_n, m\in\calM\}$. Distinguished from the traditional cross-validation or cross-validation with confidence \citep{lei2020cross}, our method does not require the moment assumptions or tail behaviors of the prediction loss. To simplify technical details, our analysis of the $\wh p_{ss,m}$ assumes that the number of bootstraps $B$ is large enough to ignore the bootstrap variability. 
	
	\begin{theorem}\label{thm: ss_consistency}
		Assume that $ \min_{j\neq m} \mu_{m,j} \geq 0$, and set the type-I error in testing (\ref{RSR_SS}) to be $\alpha$. Then conditional on the fitted candidate models  $\{\wh f_n, m\in\calM\}$, 
		\beqrs
		\P(\wh p_{ss,m}\leq\alpha) \leq \alpha +C_1 (\log^{5/6}N) n^{-1/6}+  C  n^{-K_1}(\log p)^7+ N^{-1}.
		\eeqrs
		$C$, $C_1$ are constants. $N = \max\{n, M-1\}$, where $M = |\calM|$ is the size of candidate models.
	\end{theorem}
	
	Theorem \ref{thm: ss_consistency} shows that by specifying a  type-I error rate $\alpha$, the proposed method can include the best candidate models with probability no less than $1-\alpha$ when the sample size is large enough. The error control also depends on the number of models $M$. Next, we prove the consistency of the screening procedure discussed in section \ref{sec: screen}. 
	
	\begin{theorem}\label{thm: screening_consistency}
		If the confidence set is constructed by first implementing the screening step as in (\ref{screening}). Assume the same conditions as in Theorem \ref{thm: ss_consistency}, $s>0$, $M = C_1\exp(n^{C_2})$ and $C_2<1/2$, then
		\beqrs
		\P(m\in\wh\calM_{ss}) \geq 1- \alpha  + o(1).
		\eeqrs
	\end{theorem}

	The consistency of V-fold RSR can also be proved following the Stein's method based technique as the recent developments in \cite{kissel2022high, austern2020asymptotics}. However, such an extension is primarily technical in nature and detracts from the main objective of this paper, which is to introduce a robust model selection methodology for heavy-tailed data in practice. Thus we defer the theoretical discussion to a future work. Intuitively, when the sample size is large enough, the training step can estimate the candidate models accurately enough, and the dependence on the randomness of the training data will become weak. In this paper, we numerically verify that the $V$-fold RSR is applicable and performs better than the sample splitting algorithm. We set the value of V to 5 in all the examples analyzed.
	
	\section{Gaussian Approximation for High Dimensional Two Sample U-Statistics }
	
	In this section, we provide the Gaussian approximation theory for high-dimensional two-sample U-statistics.	Let $U^n = \{U_1, \dots, U_n\}$ and $V^n = \{V_1, \dots, V_n\}$ be two samples observations, which are assumed to be independent and identically distributed drawn from their respective distributions in $\mR^p$. Specifically, $U_i = (U_{i1}, \dots, U_{ip} )\trans$, $V_i = (V_{i1}, \dots, V_{ip})\trans$. $U_i$ and $V_i$ may be dependent.
	Let $h$ be a mapping from the domain of $U\times V$ to $\mR^p$. Define the $p$-dimensional vector of U-statistics,
	\beqrs
	T =  \sqrt n\Big( \frac{1}{n(n-1)}\sum_{1\leq k\neq l\leq n}h(U_{k}, V_{l}) - \E h\Big).
	\eeqrs
	To write the U-statistics more explicitly, we note that the $j$-th element of $T$ is given by $T_j$, which is defined as
	\beqrs
	T_j = \sqrt n\Big( \frac{1}{n(n-1)}\sum_{1\leq k\neq l\leq n}h_j(U_{kj}, V_{lj}) - \E h_j\Big).
	\eeqrs
	where $h_j: \mR\times\mR\rightarrow\mR$ is the kernel of a U-statistics.
	We also calculate the projections of the U-statistics. Let
	\beqr\nonumber
	g_{1j}(u) &=& \E h_j(u, V_{1j}) - \E h_j(U_{1j}, V_{1j}), \\\nonumber
	g_{2j}(v) &=& \E h_j(U_{1j}, v) - \E h_j(U_{1j}, V_{1j}), \\\label{eq: canonical}
	f_j(u, v) &=& h_j(u, v) - \E h_j(u, V_{1j})- \E h_j(U_{1j}, v) +\E h_j(U_{1j}, V_{1j}). 
	\eeqr
	And define 
	\beqrs
	L_j = \frac{1}{\sqrt n}\sum_{i=1}^{n}\left\{  g_{1j}(U_{ij}) + g_{2j}(V_{ij})\right\}, \quad W_j = \frac{\sqrt n}{n(n-1)}\sum_{1\leq k\neq l\leq n} f_j(U_{kj}, V_{lj}).
	\eeqrs
	Clearly, $T_j = L_j + W_j$. Classical U-statistics theory implies that for each $j =1,\dots, p$, $T_j\approx L_j$. Let
	$T = (T_1,\dots, T_p)\trans$, $L = (L_1,\dots, L_p)\trans$, $g_1(U_i) = (g_{11}(U_{i1}), \dots, g_{1p}(U_{ip}))\trans$, and $g_2(V_i) = (g_{21}(V_{i1}), \dots, g_{2p}(V_{ip}))\trans$. We can also define similar $p$-dimensional vectors for functions $f(\cdot)$ and $h(\cdot)$, where the $j$-th element is given by $f_j(\cdot)$ and $h_j(\cdot)$, respectively. Denote the vectorized $L$ and $W$ as
	\beqrs
	L= \frac{1}{\sqrt n}\sum_{i=1}^{n}\left\{  g_{1}(U_{i}) + g_{2}(V_{i})\right\}, \quad W = \frac{\sqrt n}{n(n-1)}\sum_{1\leq k\neq l\leq n} f(U_{k}, V_{l}),
	\eeqrs
	\cite{chernozhukov2017central} shows that the empirical distribution of $L$ behaves similarly as $Z\sim N(0, \Gamma_g)$, where $\Gamma_g$ is the positive definite covariance matrix of $ g_{1}(U_{i}) + g_{2}(V_{i})$. For a vector $\gamma$, let $\gamma^{\otimes2} = \gamma\gamma\trans$. Then we have
	\beqrs
	\Gamma_g =\E \left\{g_{1}(U_{1})^{\otimes2} + g_{2}(V_{1})^{\otimes2}  + g_{1}(U_{1}) g_{2}(V_{1})\trans +  g_{2}(V_{1}) g_{1}(U_{1})\trans\right\}.
	\eeqrs
	And its empirical estimator is
	\beqr\label{eq: gamma_hat}
	\wh \Gamma_g  = \frac{1}{n}\sum_{i=1}^{n}\left\{\wh g_{1}(U_{i})^{\otimes2} + \wh g_{2}(V_{i})^{\otimes2}  + \wh g_{1}(U_{i}) \wh g_{2}(V_{i})\trans +  \wh g_{2}(V_{i}) \wh g_{1}(U_{i})\trans\right\}, 
	\eeqr
	where
	\beqrs
	\wh g_1(u) &=& \frac{1}{n}\sum_{i=1}^{n}h(u, V_i) - \frac{1}{n^2}\sum_{i=1}^{n}\sum_{j=1}^{n}h(U_i, V_j),\\ 
	\wh g_2(v) &=& \frac{1}{n}\sum_{i=1}^{n}h(U_i, v) - \frac{1}{n^2}\sum_{i=1}^{n}\sum_{j=1}^{n}h(U_i, V_j).\\ 
	\eeqrs
	
	Define the vector norm $\|\cdot\|$ to be the maximum value of the vector. 
	\begin{theorem}\label{thm: gauss_approxi}
		Assume that the kernel function of the U-statistics is uniformly bounded, i.e.,  $|h_j|\leq B$, $j=1,\dots, d$. Let $Z\sim N(0, \Gamma_g)$, then
		\beqrs
		\sup_{t\in\mR}\big|\P(\|T\|\leq t)  - \P(\|Z\|\leq t)\big|\leq C  n^{-K_1}(\log p)^7,
		\eeqrs	
		where $C$ and $K_1 $ are positive constants specified in the proof.
	\end{theorem}
	
	Theorem \ref{thm: gauss_approxi} provides a general theory to approximate the maximum of high dimensional two-sample U-statistics, and its proof is given in the appendix. The minimum value of the high dimensional two sample U-statistics can be obtained similarly by simply taking a negative sign of the data in Theorem \ref{thm: gauss_approxi}. Although we assume that the kernel of the U-statistics is uniformly bounded, the proof can be extended to sub-Gaussian U-statistic kernels with careful bookkeeping on the tail analysis. The bounded kernel assumption simplifies the technical details and suffices the rank-sum statistics considered in the current paper.

	\section{Simulations}
	
	In this section, we illustrate the performance of the proposed method with simulated datasets. We compare the proposed RSR with the classical cross-validation (CV) and the CVC method proposed in \cite{lei2020cross}. We also consider a setting where we define $\mu_{m,j}$ as
	\beqr\label{eq: pcv}
	\P\left\{ l(\wh f_m(X), Y) <  l(\wh f_j(X), Y)\mid \wh f_m, \wh f_j \right\} - 0.5,
	\eeqr
	and also conduct a similar statistical inference based approach for model selection. Specifically, we estimate (\ref{eq: pcv}) by sample mean and implement a similar algorithm as in \cite{lei2020cross} to obtain its $p$-value. We name this method {\bf PCV} due to the {\bf p}airwise comparison with {\bf c}ross-{\bf v}alidation. 
	
	Two examples are considered: subset selection in the low dimensional linear model and tuning parameter selection for the Lasso in the high dimensional linear model. We use $t_k$ to denote $t$-distribution with $k$ degree of freedom. The simulations are repeated 100 times, and the bootstrap number $B$ is set to 500. We implement the 5-fold validation for all the methods and let $\alpha = 0.1$ for CVC, PCV, and RSR.
	
	\subsection{Case 1: robust subset selection in linear models}
	
	In this example, we demonstrate the model selection consistency of RSR in a low-dimensional linear regression with heavy-tailed data. The response variable is generated by  $Y=X\beta+\eps$, where $\eps\sim t_1$ and $X\sim t_k$, with $k=1,2,3$. We set the true model parameter as $\beta\trans = (1,0,3,4,0)$, whose first element denotes the intercept. As a result, only the second and third variable is truly related to the response.
	The candidate models for subset selection include all the 16 possible models that have the intercept term. 	Because both the design matrix and the error term are generated from heavy-tailed distributions, the ordinary least square method fails to estimate the parameters $\wh\beta$ accurately, and the comparison of the predicted loss becomes meaningless. Thus we implement the tunning free adaptive Huber regression  \citep{sun2020adaptive, wang2020new} as the regression algorithm to estimate $\wh\beta$. In other words, the Huber loss is implemented as the loss function. The method is implemented in the R package adaHuber. Four methods are compared: cross-validation, original CVC developed in \citep{lei2020cross}, PCV as described at the beginning of the simulation, and the proposed RSR.	We consider three settings, where each column of $X$ is generated from $t_1$, $t_2$, and $t_3$, respectively. In each setting,  the sample size $n$ increases from 40 to 320.
	
	We report two evaluation criteria in Figure \ref{subset1}:  the number of non-rejected models in the confidence set and the rate of correct selection, which calculates the percentage of times when the true model is selected in the confidence set. Additional similar results where $\alpha = 0.05$ is summarized in the appendix. The selected model size tends to decrease as the sample size $n$ increases, which shows the large sample property of the methods. The rate of correct selection for RSR is almost 1 in all three cases, while the other three methods all have a positive probability of missing the true model. As expected, RSR can consistently select the true model with an estimated probability tending to 1. CVC tends to select too many models and still fails to include the true model consistently. The confidence set by PCV is only slightly smaller compared to RSR. But PCV also seems to miss the true model with non-vanishing probability. 
	
	\begin{figure}
		\centering
		\begin{subfigure}[b]{0.45\textwidth}
			\centering
			\includegraphics[scale=0.3]{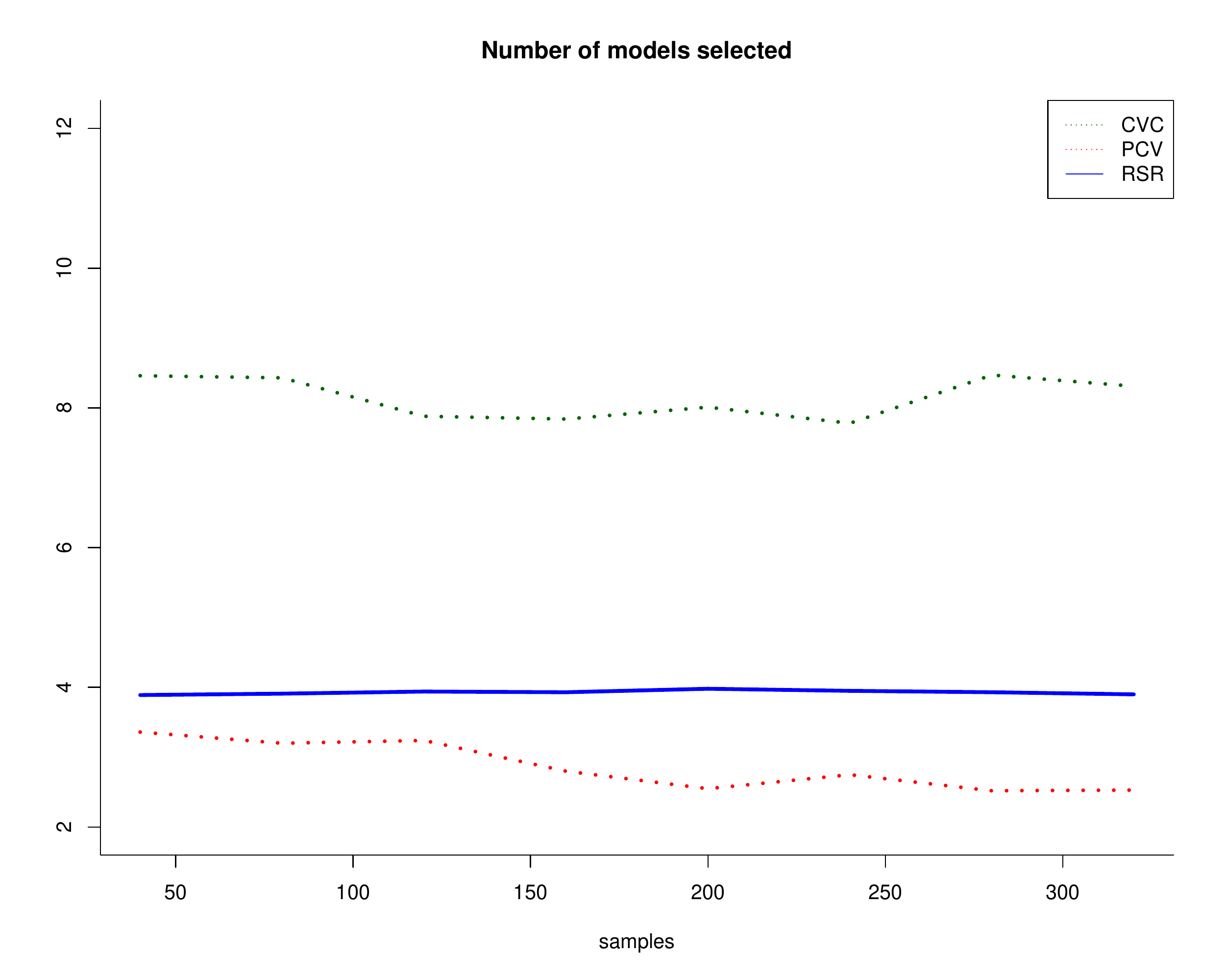}
		\end{subfigure}
		\hfill
		\begin{subfigure}[b]{0.45\textwidth}
			\centering
			\includegraphics[scale=0.3]{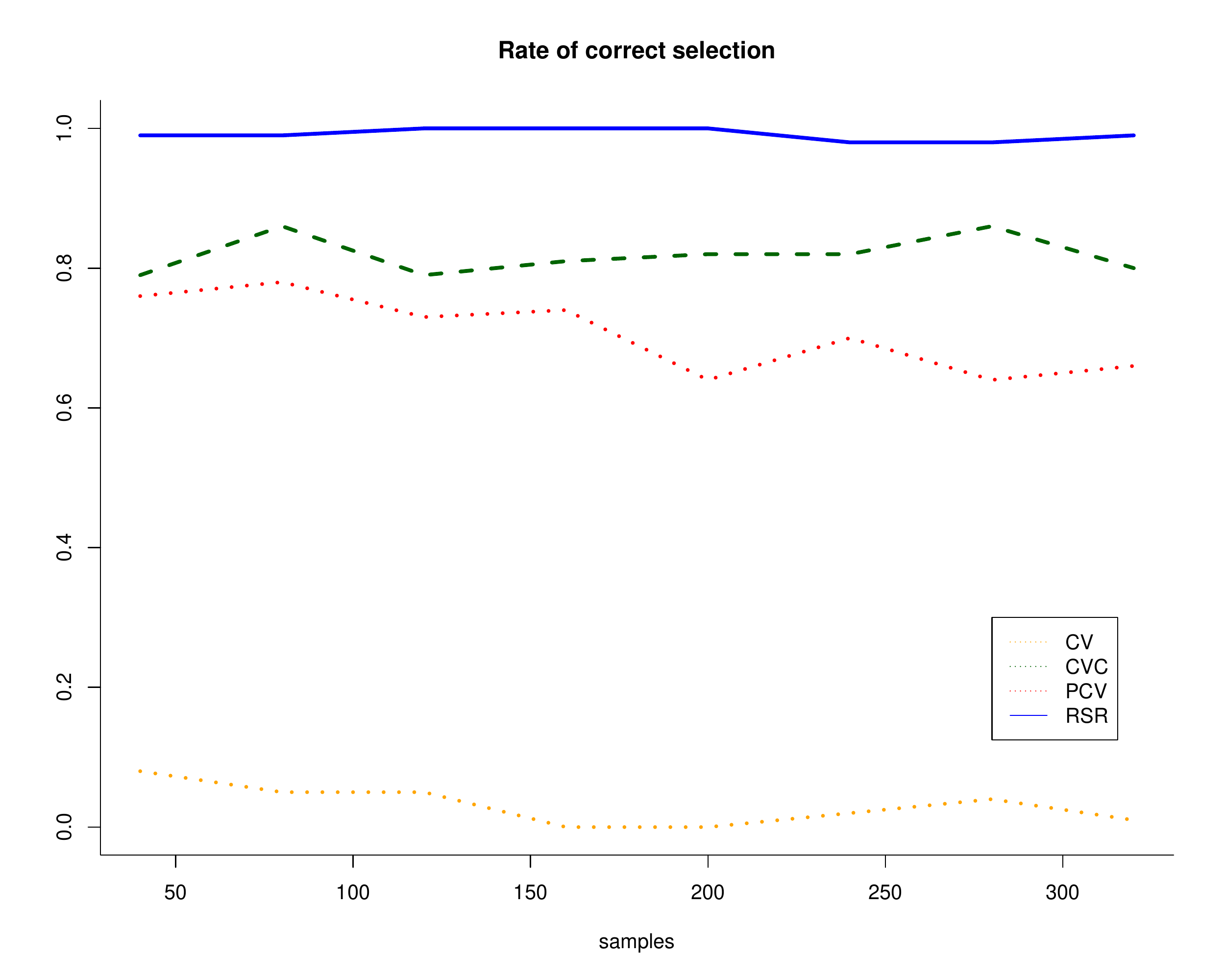}
		\end{subfigure}
		\begin{subfigure}[b]{0.45\textwidth}
			\centering
			\includegraphics[scale=0.3]{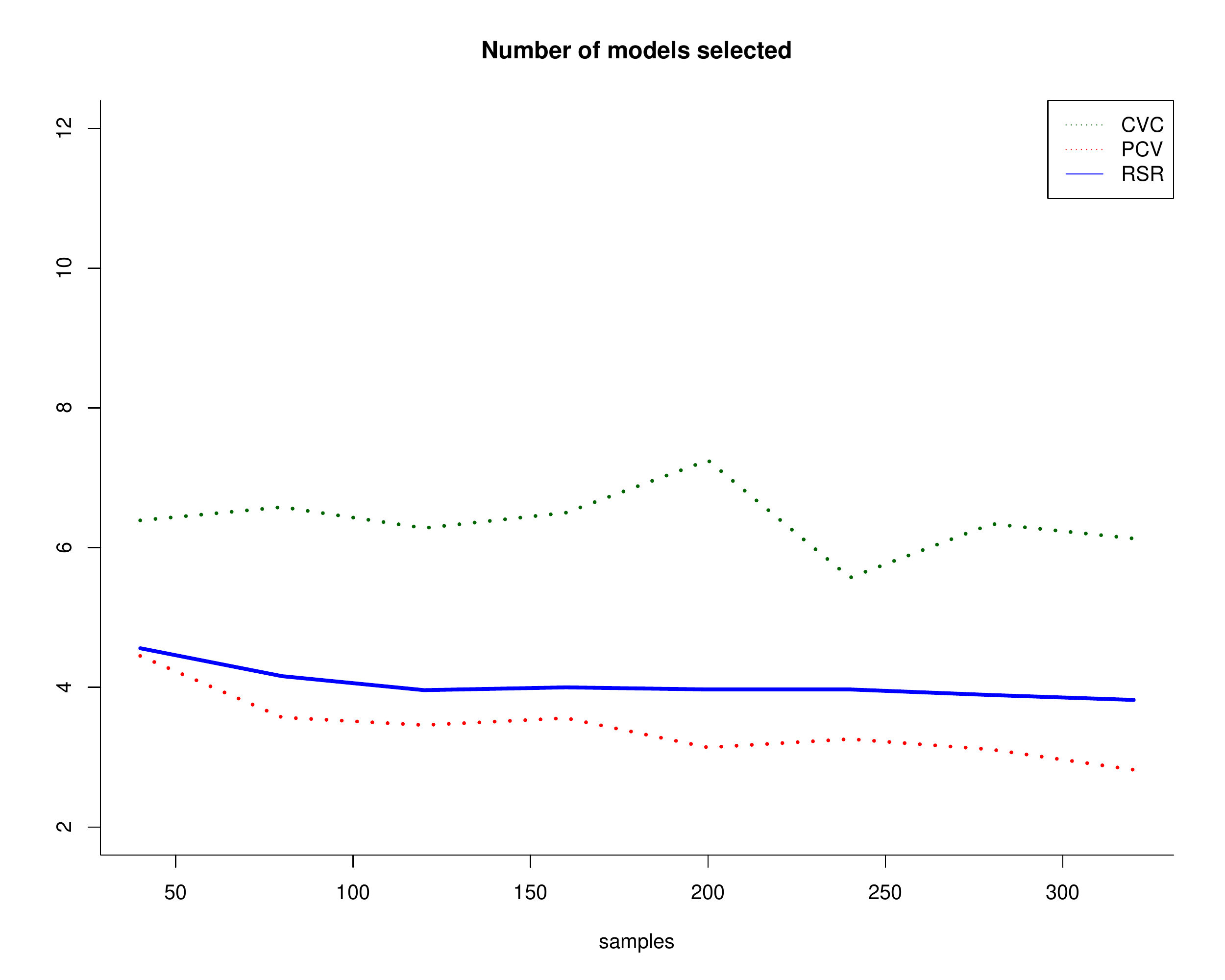}
		\end{subfigure}
		\hfill
		\begin{subfigure}[b]{0.45\textwidth}
			\centering
			\includegraphics[scale=0.3]{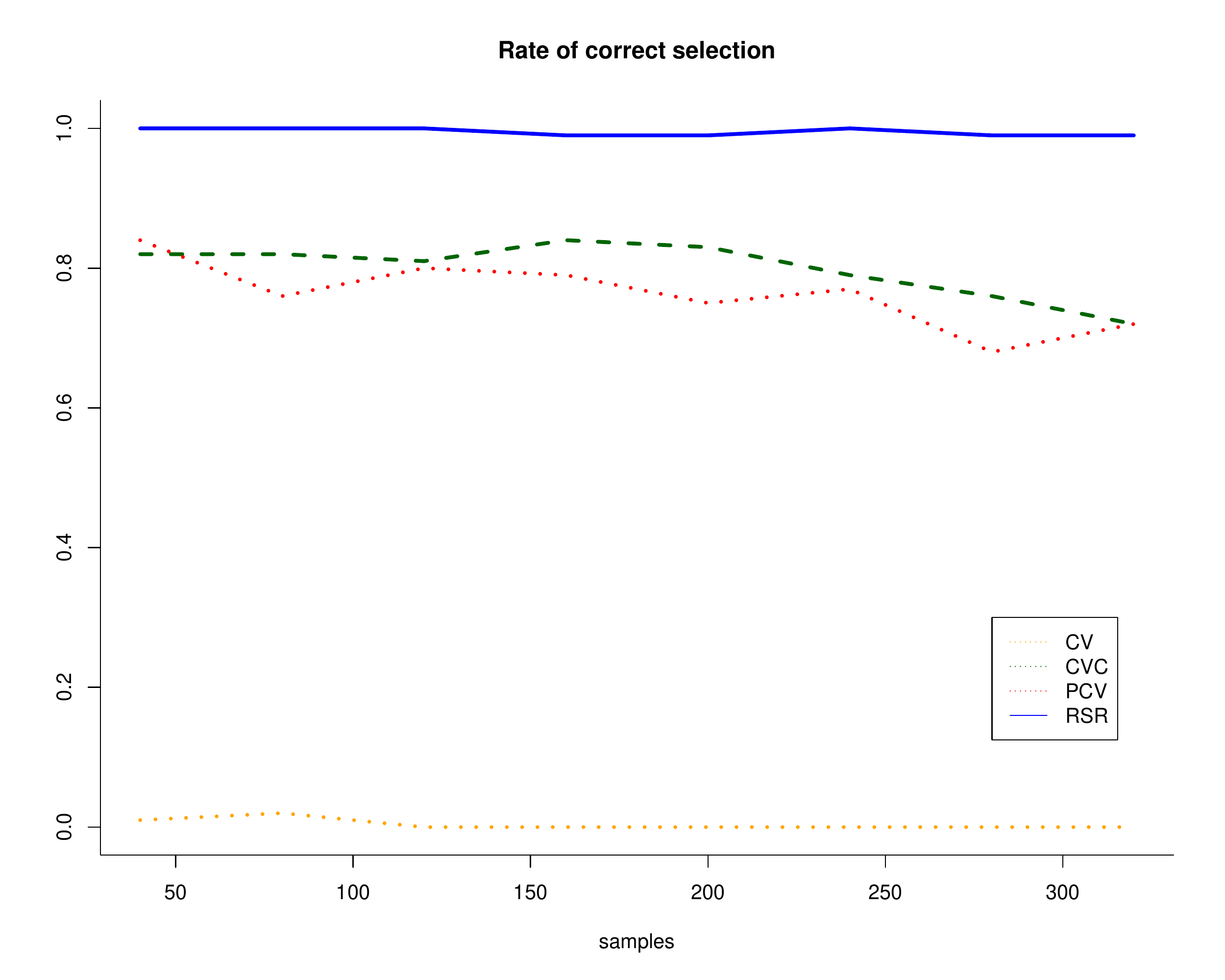}
		\end{subfigure}
		\begin{subfigure}[b]{0.45\textwidth}
			\centering
			\includegraphics[scale=0.3]{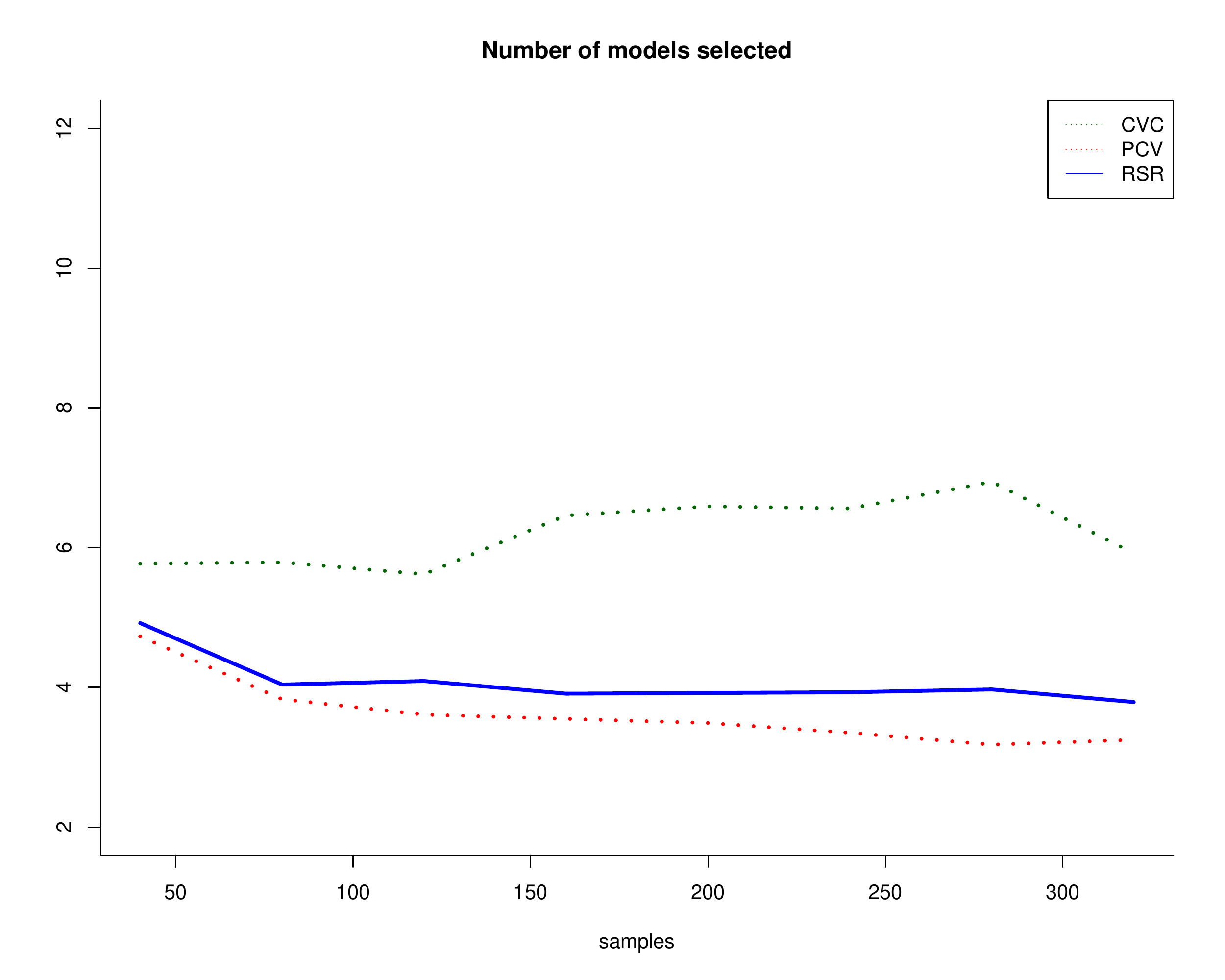}
		\end{subfigure}
		\hfill
		\begin{subfigure}[b]{0.45\textwidth}
			\centering
			\includegraphics[scale=0.3]{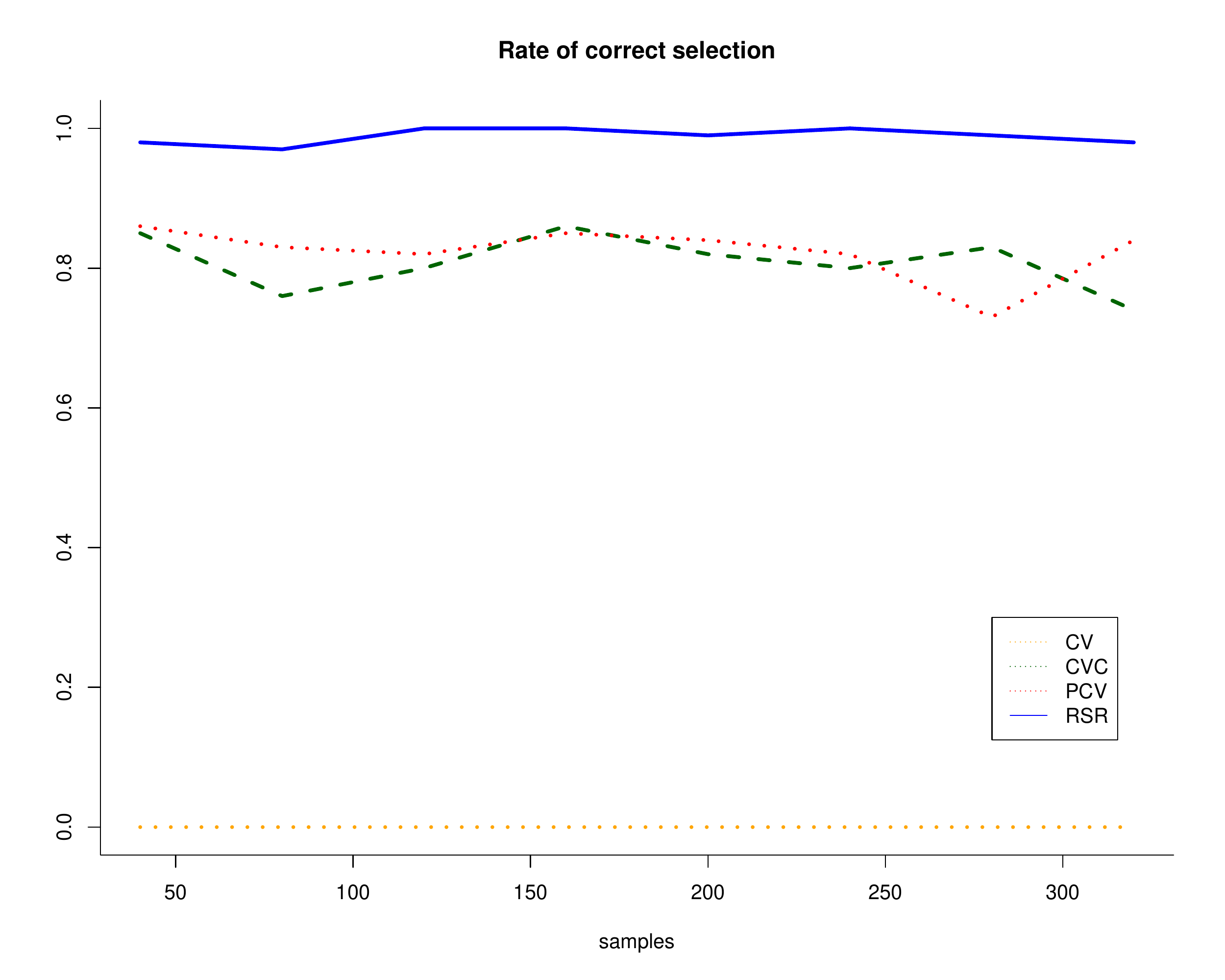}
		\end{subfigure}
		\caption{ The left column: average selected model size. The right column: the rate of correct selection. From the top row to the third row, $X$ is generated from $t_1$, $t_2$, and $t_3$, respectively. $\alpha = 0.1$   }\label{subset1}
	\end{figure}

	\subsection{Case 2: tunning the Lasso for robust risk minimization}
	
	In this example, we demonstrate the performance of RSR in choosing the Lasso tunning parameter for robust high dimensional linear regression. The data is also generated from a linear regression model, with $X\in\mR^d$ following a multivariate Gaussian distribution $N(0, \Sigma)$ and $\eps$ is an independent noise generated from $t_2$ or $t_3$. $\Sigma$ has the auto-correlation structure and $\Sigma_{ij} = \rho^{|i-j|}$, with $\rho$ taking values in $\{0.25, 0.50, 0.75\}$. 
	Two kinds of sample size and dimensionality are considered: a moderate dimensional scenario with $n=200$, $p=200$ and a high dimensional scenario with $n=400$, $p=2000$. The true model is set as $\beta^* = (1, 1, 0, 0, 0, 1, 1, 0, \dots, 0)\in\mR^d$. Thus the oracle model only contains four variables of $X$.
	
	To implement the regression algorithm, we still use the tunning-free adaptive Huber regression. For each generated dataset, we first obtain a sequence of 50 tunning parameters $\lambda$ using the R package adaHuber, then compare all four methods in selecting the $\lambda$. Because CVC, PCV, and RSR tend to select multiple $\lambda$, we always choose the largest $\lambda$ in the confidence set to obtain the sparsest model. Note that when using $K$ fold validation, $\lambda$ is trained on the dataset with sample size $(K-1)n/K$. It has been widely studied in the literature that the optimal $\lambda$ is inverse proportional to the square root of the sample size \citep{van2008high, meinshausen2009lasso, dalalyan2017prediction}. Thus we correct the estimated $\lambda$ by $ \sqrt{1 - K^{-1}}\wh\lambda$ for all the four methods, then apply the corrected $\lambda$ on the entire dataset to obtain the estimated $\wh\beta$. 
	
	We first report the average size of the number of non-zero values in $\wh\beta$ in Table \ref{beta_avg_size}. We can see that the average size of the selected models for RSR is not only the smallest but also the closest to 4, which is the true model size. This indicates that RSR tends to output a sparser model compared to other methods. The immediate question is, are those selected variables accurate? Or are they only random four variables in the design matrix? We summarize the percentage of times that $\wh\beta$ has non-zero values on the support of $\beta^*$ in Table \ref{beta_include_percent}. Clearly, all four methods can identify the true variables with high probability. But the models selected by CV, CVC, and PCV have too many false positives. Note that the estimated inclusion probability in Table \ref{beta_include_percent} is not the same as the high probability bound ($1-\alpha$) derived in the theorems or the rate of correct selection in Figure \ref{subset1}. The estimated probability in Table \ref{beta_include_percent} is only based on the largest $\lambda$ in the confidence set selected by CVC, PCV, and RSR. We further demonstrate the accuracy of RSR in Table \ref{beta_oracle_rate}, where we calculate the percentage of times that $\wh\beta$ exactly selects the four true variables in $\beta^*$. We can see that the proposed method has a huge advantage in identifying the exact oracle model. The other methods such as cross-validation tend to be conservative and select too many variables into the model. Lastly, we report the cross-validated error in Table \ref{cv_error}. We see that RSR has slightly higher predictive error compared to other methods. It further demonstrates the advantage of the proposed method in model selection: the model selected by RSR may not have the smallest predictive mean square error, but it is able to identify the true variables accurately. In other words, the widely used cross-validated error is not necessarily the best choice for robust model selection.

	\begin{table}[ht]
		\centering\caption{The average number of non-zero values in $\wh\beta$, which is estimated by the selected $\lambda$. For CVC, PCV, and RSR, we choose the largest $\lambda$ in the confidence set. }\label{beta_avg_size}
		\begin{tabular}{ccccccc}
			\hline
			$(n, p)$ & $df$ & $\rho$ & CV & CVC & PCV & RSR \\ 
			\hline
			\multirow{6}*{$(200, 200)$} &  \multirow{3}*{2}  & 0.25 & 9.47 & 8.42 & 8.07 & 4.47 \\ 
			&  & 0.50 & 9.15 & 7.69 & 7.77 & 4.55 \\ 
			&  & 0.75 & 8.78 & 7.18 & 7.19 & 4.83 \\ 
			&  \multirow{3}*{3}  & 0.25 & 10.13 & 8.92 & 7.68 & 4.35 \\ 
			&  & 0.50 & 8.86 & 7.67 & 6.68 & 4.26 \\ 
			&  & 0.75 & 8.87 & 7.13 & 6.40 & 4.65 \\ \hline
			\multirow{6}*{$(400, 2000)$} &  \multirow{3}*{2}  & 0.25 & 11.64 & 13.28 & 12.20 & 4.56 \\ 
			&  & 0.50 & 10.88 & 11.92 & 10.94 & 4.54 \\ 
			&  & 0.75 & 9.15 & 9.45 & 9.09 & 4.61 \\ 
			&  \multirow{3}*{3}  & 0.25 & 12.39 & 13.36 & 10.72 & 4.15 \\ 
			&  & 0.50 & 10.92 & 11.78 & 10.13 & 4.26 \\ 
			&  & 0.75 & 10.11 & 9.19 & 7.94 & 4.54 \\ 
			\hline
		\end{tabular}
	\end{table}
	
	\begin{table}[ht]
		\centering\caption{The estimated probability of the support of $\wh\beta$ covers the support of $\beta^*$. $\wh\beta$ is estimated by plugging in the  $\lambda$ selected by each method. }\label{beta_include_percent}
		\begin{tabular}{ccccccc}
			\hline
			$(n, p)$ & $df$ & $\rho$ & CV & CVC & PCV & RSR \\ 
			\hline
			\multirow{6}*{$(200, 200)$} &  \multirow{3}*{2} & 0.25 & 1.00 & 0.97 & 1.00 & 0.94 \\ 
			&  & 0.50 & 0.99 & 0.96 & 0.99 & 0.99 \\ 
			&  & 0.75 & 0.99 & 0.96 & 0.98 & 0.94 \\ 
			&  \multirow{3}*{3} & 0.25 & 1.00 & 1.00 & 1.00 & 0.99 \\ 
			&  & 0.50 & 1.00 & 1.00 & 1.00 & 1.00 \\ 
			&  & 0.75 & 1.00 & 1.00 & 1.00 & 1.00 \\ \hline
			\multirow{6}*{$(400, 2000)$} &  \multirow{3}*{2}  & 0.25 & 1.00 & 1.00 & 1.00 & 1.00 \\ 
			&  & 0.50 & 1.00 & 1.00 & 1.00 & 1.00 \\ 
			&  & 0.75 & 1.00 & 1.00 & 1.00 & 1.00 \\ 
			&  \multirow{3}*{3} & 0.25 & 1.00 & 1.00 & 1.00 & 1.00 \\ 
			&  & 0.50 & 1.00 & 1.00 & 1.00 & 1.00 \\ 
			&  & 0.75 & 1.00 & 1.00 & 1.00 & 1.00 \\ 
			\hline
		\end{tabular}
	\end{table}

	\begin{table}[ht]
		\centering\caption{The estimated probability of the support of $\wh\beta$ is exactly the same as the support of $\beta^*$. $\wh\beta$ is estimated by plugging in the  $\lambda$ selected by each method.  }\label{beta_oracle_rate}
		\begin{tabular}{ccccccc}
			\hline
			$(n, p)$ & $df$ & $\rho$ & CV & CVC & PCV & RSR \\ 
			\hline
			\multirow{6}*{$(200, 200)$}  &  \multirow{3}*{2}  & 0.25 & 0.02 & 0.10 & 0.02 & 0.61 \\ 
			&  & 0.50 & 0.05 & 0.08 & 0.03 & 0.67 \\ 
			&  & 0.75 & 0.04 & 0.08 & 0.08 & 0.34 \\ 
			&  \multirow{3}*{3}  & 0.25 & 0.00 & 0.01 & 0.11 & 0.77 \\ 
			&  & 0.50 & 0.02 & 0.10 & 0.16 & 0.82 \\ 
			&  & 0.75 & 0.02 & 0.12 & 0.19 & 0.58 \\ \hline
			\multirow{6}*{$(400, 2000)$} &  \multirow{3}*{2}  & 0.25 & 0.04 & 0.05 & 0.03 & 0.64 \\ 
			&  & 0.50 & 0.01 & 0.04 & 0.01 & 0.66 \\ 
			&  & 0.75 & 0.06 & 0.05 & 0.01 & 0.53 \\ 
			&  \multirow{3}*{3}  & 0.25 & 0.03 & 0.01 & 0.04 & 0.90 \\ 
			&  & 0.50 & 0.04 & 0.06 & 0.06 & 0.81 \\ 
			&  & 0.75 & 0.02 & 0.05 & 0.07 & 0.62 \\ 
			\hline
		\end{tabular}
	\end{table}

	\begin{table}[h]
		\centering
		\caption{\label{err} The average cross-validated prediction error for all the methods, based on the  $\lambda$ selected by each method. }\label{cv_error}
		\begin{tabular}{ccccccc}
			\hline
			$(n, p)$ & $df$ & $\rho$ & CV & CVC & PCV & RSR \\ 
			\hline
			\multirow{6}*{$(200, 200)$}  & \multirow{3}*{2} & 0.25 & 6.89 & 6.97 & 7.00 & 7.46 \\ 
			&  & 0.50 & 6.46 & 6.51 & 6.56 & 6.95 \\ 
			&  & 0.75 & 5.32 & 5.38 & 5.43 & 5.86 \\ 
			& \multirow{3}*{3} & 0.25 & 2.73 & 2.75 & 2.80 & 3.12 \\ 
			&  & 0.50 & 2.77 & 2.78 & 2.80 & 3.16 \\ 
			&  & 0.75 & 2.79 & 2.79 & 2.81 & 3.14 \\ \hline
			\multirow{6}*{$(400, 2000)$}  & \multirow{3}*{2} & 0.25 & 6.32 & 6.34 & 6.31 & 6.65 \\ 
			&  & 0.50 & 6.70 & 6.69 & 6.69 & 7.03 \\ 
			&  & 0.75 & 6.89 & 6.87 & 6.93 & 7.25 \\ 
			& \multirow{3}*{3} & 0.25 & 2.63 & 2.65 & 2.66 & 2.89 \\ 
			&  & 0.50 & 2.82 & 2.82 & 2.87 & 3.01 \\ 
			&  & 0.75 & 2.77 & 2.77 & 2.78 & 3.02 \\ 
			\hline
		\end{tabular}
	\end{table}

		\section{Real Data Analysis}
	
	The advancement technology in single-cell genomics has provided researchers with simultaneous profiling of multiple omics features in the same cell, and allows scientists to analyze diseases and immune reactions at the single-cell level. In this section, we analyze a timely single-cell multi-omics dataset measured on Covid-19 patients \citep{stephenson2021single}. In the original paper, it was found that the proportion of DC1 cells within all immune cells will change among patients with different statuses of covid-19, including healthy, asymptomatic, mild, moderate, severe, and critical. Both the high dimensional RNA-seq and surface antibodies are measured for the DC1 cells. And it is of great interest to build a prediction model to identify the potential structural relations between RNA and proteins \citep{zhou2020surface, cai2022PNAS}. In this analysis, we build a high dimensional linear model for prediction and select the tuning parameter $\lambda$ by all four methods compared in the simulation. Specifically, $X$ are the RNA-seq and $Y$ are the surface antibodies, and each single-cell corresponds to one observation.
	
	We begin with performing standard quality control and retaining the highly variable genes in the DC1 cells as annotated in the original data. Specifically, we set the minimum number of genes per cell to 200 and set the minimum number of cells per gene to 3. Then we select the highly variable genes annotated in the AnnData object, which is one of the most popular formats in Python to facilitate single-cell data analysis. We obtained 2376 genes and 47 proteins measured on 366 cells. Both the protein and RNA measurements entail heavy-tailed distributions as shown in Figure \ref{fig: des_real_data}. 	We implement the truncated lasso \citep{fan2016shrinkage} as the regression algorithm due to its robustness in genomics data analysis and select the tuning parameter $\lambda$ by all the methods compared in the simulation. Two scenarios are considered: 1) we implement the original Lasso without any truncation on the data, and 2) both $X$ and $Y$ are truncated on the 99\% quantile. That is, the values higher than 99\% of $X$ or $Y$ are set to be the 99\% quantile. The boxplot of the cross-validated error and selected model size of the four methods are reported in Figure \ref{fig: real1}, where $\alpha$ is set as 0.05. Additional numerical results where $\alpha=0.01$ is provided in the appendix. Similar as the simulations, we select the largest $\lambda$ in the confidence set to obtain the sparsest model. We observe in Figure \ref{fig: real1} that the cross-validated prediction errors are all very similar for all the methods. However,  the candidate models selected by RSR have a significantly smaller model size. It shows the RSR is able to identify the most informative subset of genes in predicting the surface proteins, even if the single-cell data has heavy-tail distributions.
	
	\begin{figure}
		\centering
		\begin{subfigure}[b]{0.45\textwidth}
			\centering
			\includegraphics[scale=0.35]{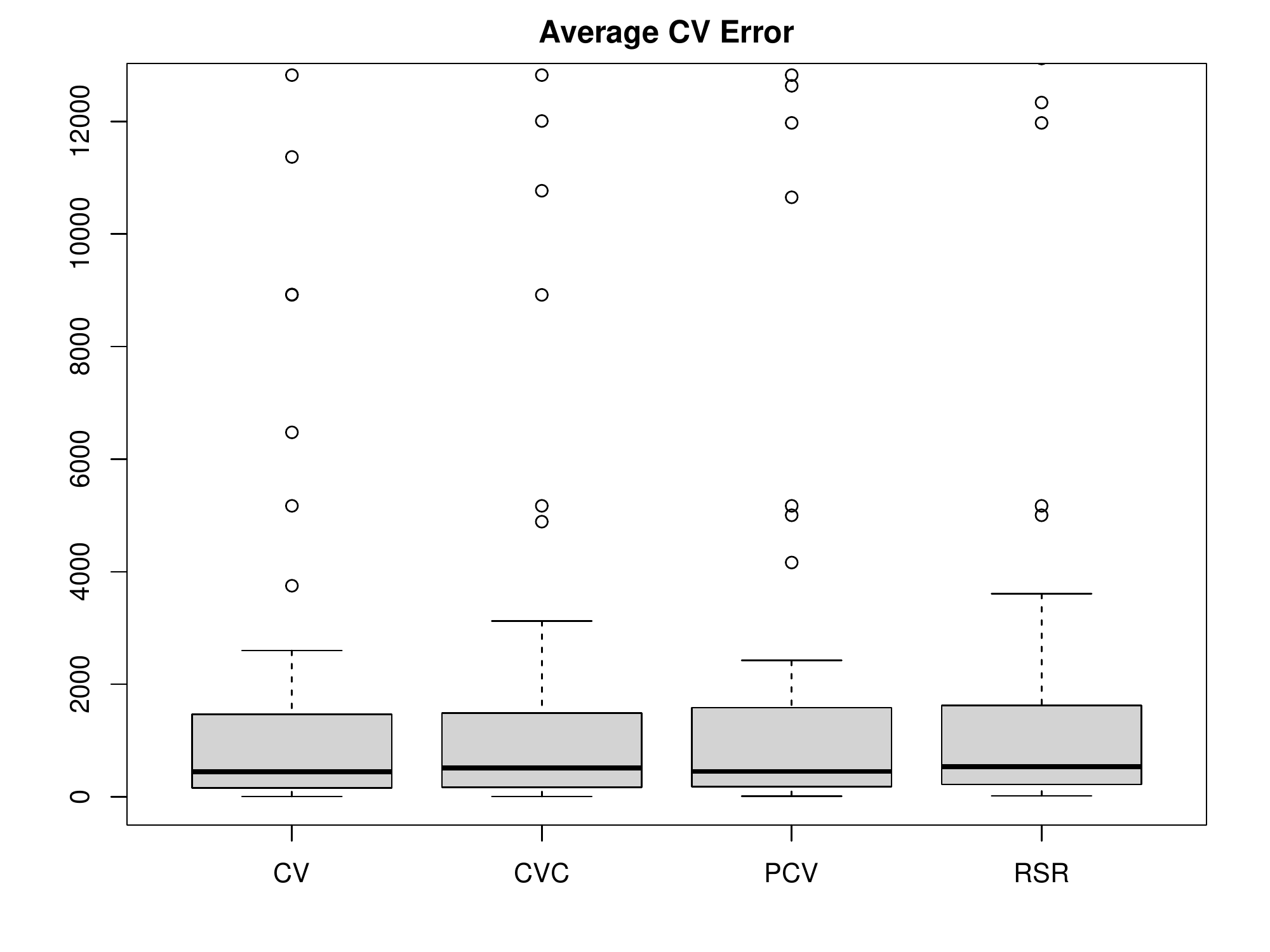}
		\end{subfigure}
		\hfill
		\begin{subfigure}[b]{0.45\textwidth}
			\centering
			\includegraphics[scale=0.35]{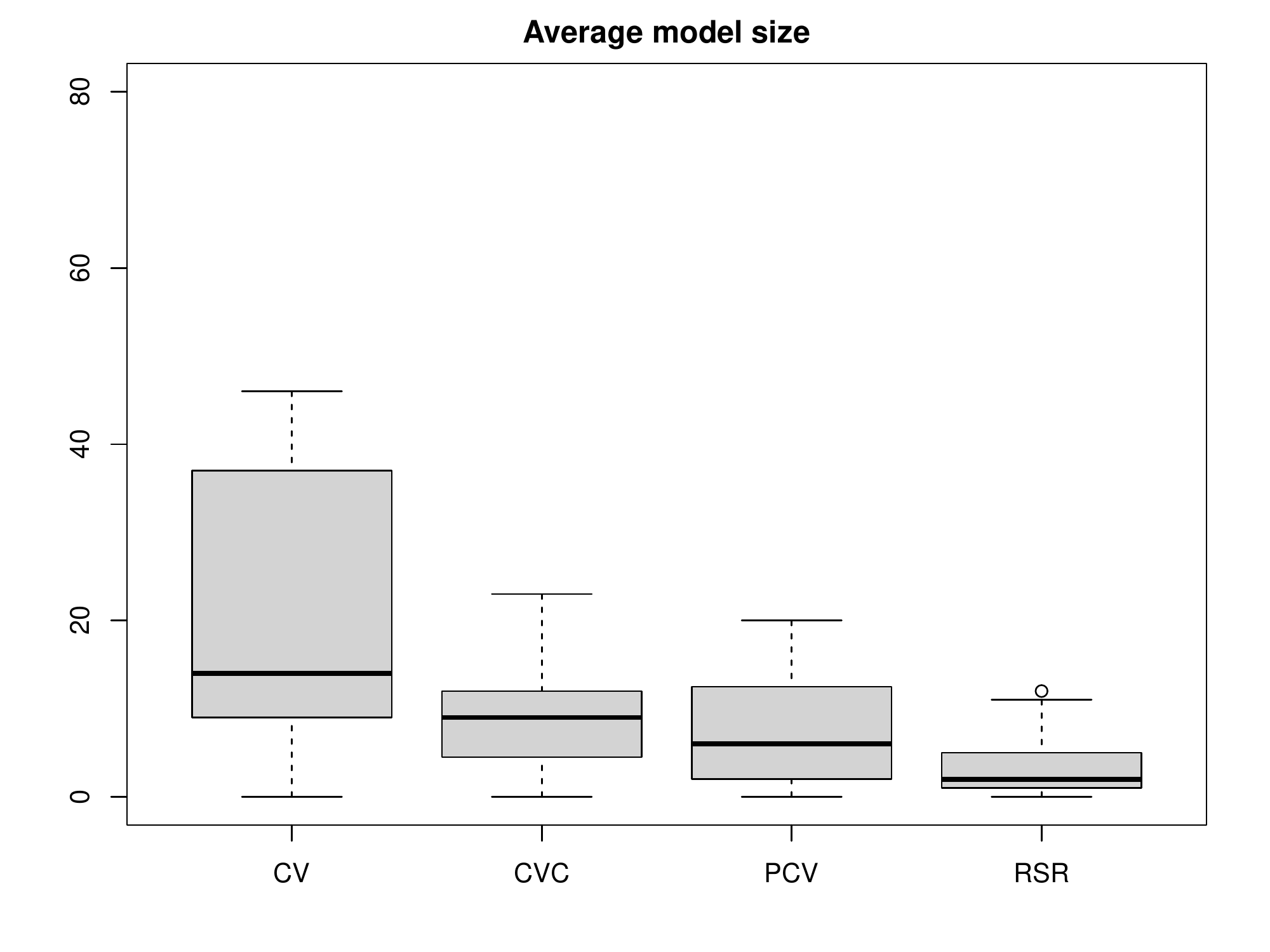}
		\end{subfigure}
		\begin{subfigure}[b]{0.45\textwidth}
			\centering
			\includegraphics[scale=0.35]{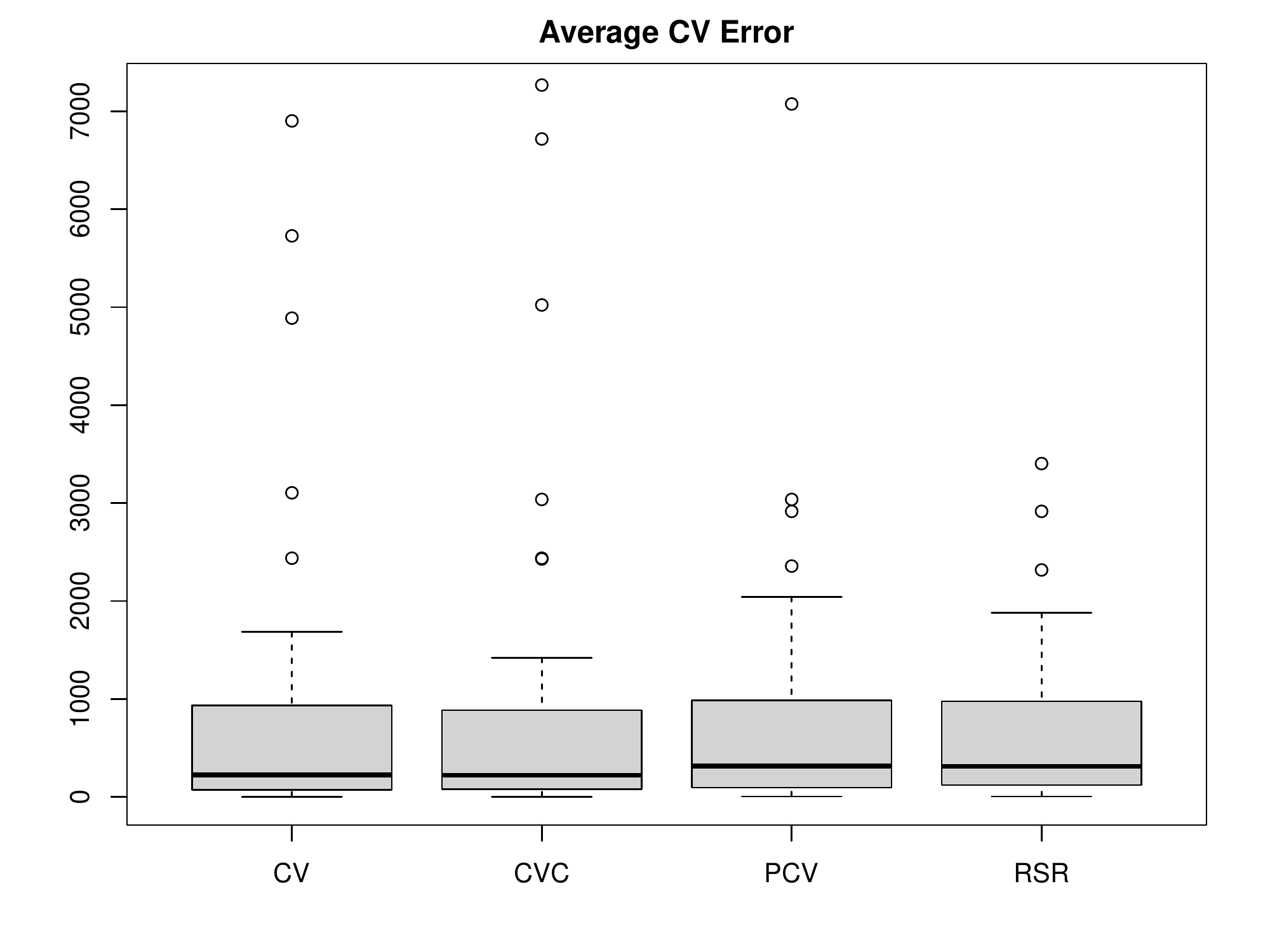}
		\end{subfigure}
		\hfill
		\begin{subfigure}[b]{0.45\textwidth}
			\centering
			\includegraphics[scale=0.35]{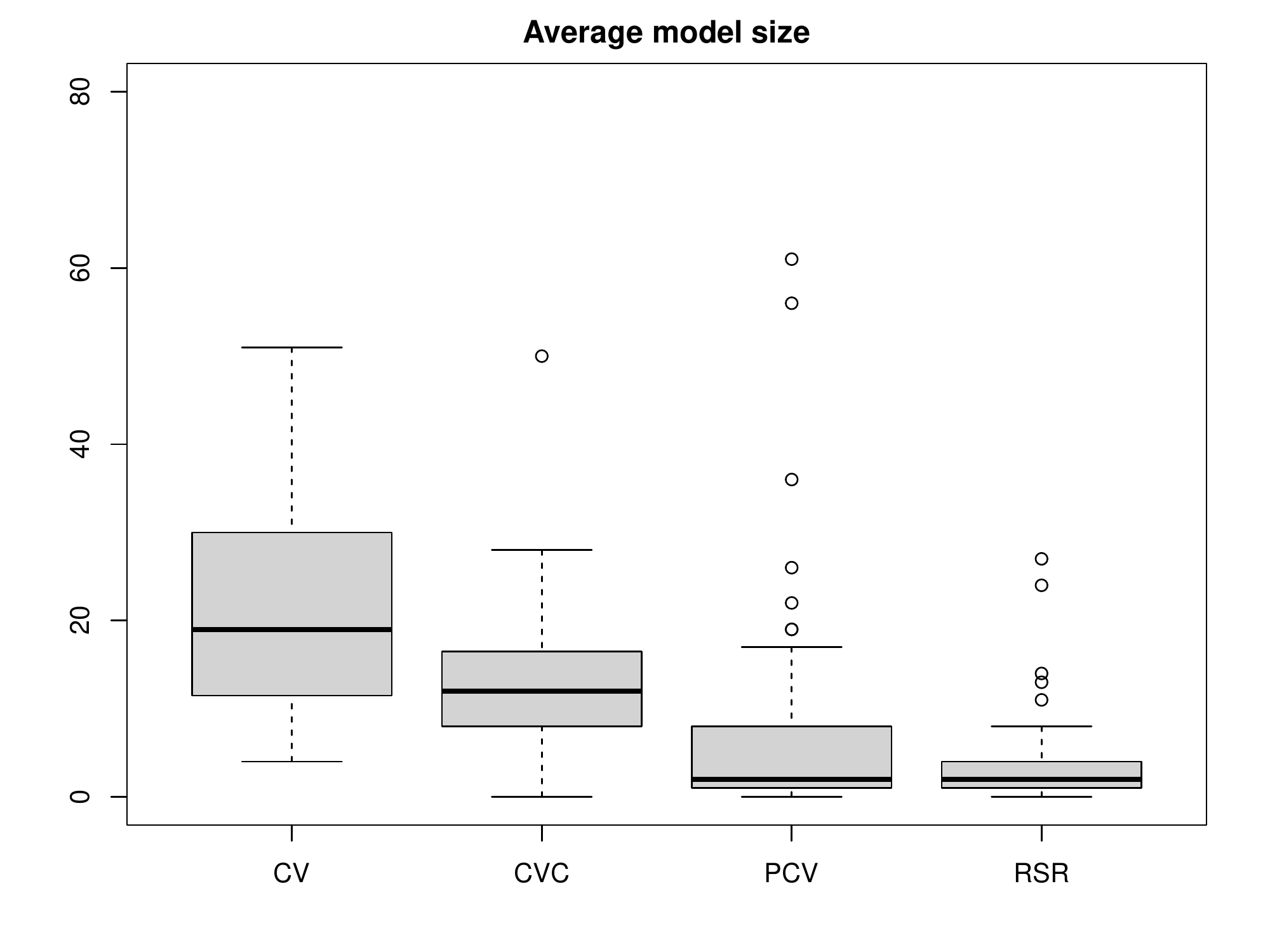}
		\end{subfigure}
		\caption{ Test error and selected model size. The first row represents Lasso without truncation on the data. The second row represents Lasso with truncation on the 99\% quantile of both $X$ and $Y$. $\alpha = 0.05$.}\label{fig: real1}
	\end{figure}

	\section{Discussion}

	In this paper, we propose a robust rank-sum based approach for general model selection problems in heavy-tailed data. By simultaneously testing the distribution of the generalized rank-sum statistics, we transform the cross-validation problem into a high-dimensional testing problem and provide its consistency in selecting optimal candidate models. We also develop a general theory to approximate the minimum/maximum value of a high dimensional vector of two sample U-statistics, with potential applications in other related statistical areas.
	
	We briefly discuss several areas of topics left for future research. It is of interest to prove consistency for model selections in particular statistical problems, such as robust high dimensional linear regression, matrix recovery, principal component analysis, etc., under heavy-tailed data. Besides, we are interested in developing tools to characterize the dependence among samples in the V-fold RSR. Similar techniques can be extended in other sample splitting based high dimensional inference problems. Last but not least, we advocate the key intuition of this paper: apply the idea of statistical inference to other related problems, such as cross-validation in this paper. Similar ideas have been popular in recent years and have already produced fruitful research products, such as the Gaussian differential privacy proposed in \cite{dong2022gaussian}.

\bibliography{cairui}		
\bibliographystyle{refstyle}

	\appendix

\section{Proofs}

For any covariance matrix $\Sigma$, let $Z_{\Sigma}$ denote the random vector drawn from $N(0, \Sigma)$. For $\alpha\in(0,1)$, let $z(\alpha, \Sigma)$ be the $\alpha$-quantile of the minimum of $Z_{\Sigma}$. Denote $M = |\calM|$, $p= M-1$ and $N = \max\{n, p\}$.  $c_1$, $c_2$, $\dots$ and $C_1$, $C_2$, $\dots$ denote constants whose values might vary depending on the situation.

\subsection{Technical Lemmas}



\begin{lemma}\label{lem: mill_ratio}
	Let $\phi(x)$ and $\Phi(x)$ be the density function and cumulative distribution function of standard normal distribution. Then for arbitrary $x>0$, the following inequalities hold
	\beqrs
	\frac{x}{\phi(x)}\leq \frac{1}{1-\Phi(x)}\leq \frac{x}{\phi(x)}\frac{1+x^2}{x^2}.
	\eeqrs
	In particular,
	\beqrs
	\int_{x}^{\infty}\exp(-\frac{u^2}{2})  du \approx  \frac{1}{x}\exp(-\frac{x^2}{2})
	\eeqrs
	holds when $x\rightarrow\infty$.
\end{lemma}
Lemma \ref{lem: mill_ratio} is known as the Mill's Ratio inequality. See for example, Theorem 1 in \cite{fan2012new}. 

\begin{lemma}\label{lem: Fuk-Nagaev}
	Let $X_1,\dots, X_n$ be independent random vectors in $\mR^p$. Define $\sigma^2 = \max_{1\leq j\leq p}\sum_{i=1}^{n}\E[X_{ij}^2]$. Then for every $s>2$ and $t>0$,
	\beqrs
	&&\P\Big(\max_{1\leq j\leq p} |\frac{1}{n}\sum_{i=1}^{n}(X_{ij} - \E[X_{ij}])|\geq 2\E\Big[ \max_{1\leq j\leq p} |\frac{1}{n}\sum_{i=1}^{n}(X_{ij} - \E[X_{ij}])|\Big] + t \Big)\\
	&&\leq \exp(-\frac{n^2t^2}{3\sigma^2}) + \frac{K_s}{n^st^s}\sum_{i=1}^{n}\E\Big[ \max_{1\leq j\leq p}|X_{ij}|^s\Big]
	\eeqrs
\end{lemma}
Lemma \ref{lem: Fuk-Nagaev} is a Fuk-Nagaev type inequality. See, for example, Theorem 3.1 in \cite{einmahl2008characterization} and Lemma D.2 in \cite{chernozhukov2019inference}.

\begin{lemma}\label{lem: sigma_bound}
	Assume that $p = C_1\exp(n^{C_2})$ and $C_2<1/2$, then
	\beqrs
	\P\Big(\max_{1\leq j\leq p}|\frac{\sigma_j}{\wh\sigma_j} - 1|>r \Big) = o(1),
	\eeqrs
	where $r = c_1 n^{-c_2}$,  $0<C_2<c_2<1-C_2$. $r\in(0, 1)$ when $n$ is sufficiently large.
\end{lemma}

\begin{proof}
	
	Note that for $a>0$, $|\sqrt a - 1|= |a-1|/(\sqrt a + 1)\leq |a-1|$, we have
	\beqrs
	\P\Big(\max_{1\leq j\leq p}|\frac{\sigma_j}{\wh\sigma_j} - 1|>r \Big)\leq 	\P\Big(\max_{1\leq j\leq p}|\frac{\sigma_j^2}{\wh\sigma_j^2} - 1|>r \Big).
	\eeqrs
	Let $t = n^{-c_3}$, $c_3>0$ such that $0<c_2+c_3<1-C_2$.  Recall that  $r = c_1 n^{-c_2}$. Then 
	\beqrs
	\P\Big(\max_{1\leq j\leq p}|\frac{\sigma_j^2}{\wh\sigma_j^2} - 1|>r \Big)&=& \P\Big(\max_{1\leq j\leq p}|\frac{\sigma_j^2}{\wh\sigma_j^2} - 1|>r, \min_j\wh\sigma_j^2> t \Big) + \P\Big(\max_{1\leq j\leq p}|\frac{\sigma_j^2}{\wh\sigma_j^2} - 1|>r, \min_j\wh\sigma_j^2\leq t \Big)\\
	&\leq&\P\Big(\max_{1\leq j\leq p}|\wh\sigma_j^2 - \sigma_j^2 |> rt \Big) + \P\Big(\min_j\wh\sigma_j^2\leq t \Big)
	\eeqrs
	By the union bound and Bernstein inequality, and $p = C_1\exp(n^{C_2})$, 
	\beqrs
	\P\Big(\min_j\wh\sigma_j^2\leq t \Big) \leq \sum_{j=1}^{p} \P(\wh\sigma_j^2\leq t)\leq p\exp( - C_3 t^2 n) = C_1\exp(-C_3 n^{1-2c_3} + n^{C_2}) , 
	\eeqrs
	which goes to zero when $C_2<1$.
	Recall that
	\beqrs  
	&&\sigma_j^2 = \frac{1}{6n} - 2 \cov\{F_{1}(B_1),  F_{2,j}(A_1)\},\\ 
	&&\wh\sigma_{j}^2 =  \frac{1}{6n} - \frac{2}{n}\Big( \sum_{i=1}^{n}\wh F_{1}(B_i) \wh F_{2, j}(A_i)  - \sum_{i=1}^{n} \wh F_{1}(B_i) \sum_{i=1}^{n}\wh F_{2,j}(A_i)     \Big).
	\eeqrs
	Thus we have
	\beqrs
	&&\P\Big(\max_{1\leq j\leq p}| \sigma_j^2 -\wh\sigma_j^2 |>rt \Big)\\
	&\leq&\P\Big(\max_{1\leq j\leq p}|\frac{1}{n}\sum_{i=1}^{n}\wh F_{1}(B_i) \wh F_{2, j}(A_i)  - \E(F_{1}(B_1) F_{2, j}(A_1)) | > rt/4 \Big) \\
	&&+ \P\Big( \max_{1\leq j\leq p}|\frac{1}{n} \sum_{i=1}^{n}\wh F_{1}(B_i)\frac{1}{n}\sum_{i=1}^{n}\wh F_{2,j}(A_i) - \E F_{1}(B_1)\E F_{2,j}(A_1)| >{ rt/4}\Big) \\
	&\defn& D_1+D_2.
	\eeqrs
	Furthermore,
	\beqrs
	D_2 \leq&& \P\Big( \max_{1\leq j\leq p}|\frac{1}{n} \sum_{i=1}^{n}\wh F_{1}(B_i)\big(\frac{1}{n}\sum_{i=1}^{n}\wh F_{2,j}(A_i) - \E F_{2,j}(A_1)\big)| >{ rt/8}\Big) + \\
	&&\P\Big( \max_{1\leq j\leq p}|\big(\frac{1}{n} \sum_{i=1}^{n}\wh F_{1}(B_i) - \E F_{1}(B_1)\big)\E F_{2,j}(A_1)| >{ rt/8}\Big).  
	\eeqrs
	Note that $\E F_{2,j}(A_1) = 1/2$, and $n^{-1}\sum_{i=1}^{n}\wh F_{1}(B_i)$ goes to $1/2$ when $n$ is large enough. Thus it suffices to bound
	\beqrs
	\P\Big( \max_{1\leq j\leq p}|\frac{1}{n}\sum_{i=1}^{n}\wh F_{2,j}(A_i) - \E F_{2,j}(A_1)| > rt/4\Big) + \P\Big( \max_{1\leq j\leq p}|\frac{1}{n} \sum_{i=1}^{n}\wh F_{1}(B_i) - \E F_{1}(B_1)| > rt/4\Big).  
	\eeqrs
	Because  $\wh F_{1}$, $\wh F_{2,j}$, $ F_{1}$, $ F_{2,j}$ are all uniformly bounded on $[0, 1]$. By Lemma \ref{maxE_bound}, we have
	\beqrs
	\E\Big[\max_{1\leq j\leq p}|\frac{1}{n}\sum_{i=1}^{n}\wh F_{2,j}(A_i)   - \E F_{2,j}(A_1) | \Big] \leq K_1 \log(p)/\sqrt n.
	\eeqrs
	By Lemma \ref{lem: Fuk-Nagaev}, for every $t_1>0$,
	\beqrs
	&&\P\Big( \max_{1\leq j\leq p}|\frac{1}{n}\sum_{i=1}^{n}\wh F_{2,j}(A_i)   - \E F_{2,j}(A_1) | >t_1 +2 K_1\log(p)/ n\Big)\\
	\leq&& \exp(-Cn t_1^2) + C t_1^{-2}n^{-1}.
	\eeqrs
	Let $t_1 = n^{-c_4}$ such that $1-2c_4>0$, and $c_4>c_2+c_3$. Then we have
	$t_1 +2 K_1\log(p)/ n= O(n^{-c_4} + n^{C_2-1}) < rt/4 = O(n^{-(c_2+c_3)})$, then we have
	\beqrs
	\P\Big( \max_{1\leq j\leq p}|\frac{1}{n}\sum_{i=1}^{n}\wh F_{2,j}(A_i) - \E F_{2,j}(A_1)| > rt/4\Big) < o(1).
	\eeqrs
	Similarly, we can show $\P\Big( \max_{1\leq j\leq p}|\frac{1}{n} \sum_{i=1}^{n}\wh F_{1}(B_i) - \E F_{1}(B_1)| > rt/4\Big)  = o(1)$. And $D_1 = o(1)$. This finishes the proof of Lemma \ref{lem: sigma_bound}.
	
\end{proof}

\begin{lemma}\label{maxE_bound}
	Let $X_j$, $j=1,\dots,d$, $d\geq2$, be a sequence of random variables with concentration bounds 
	\beqrs
	\P(|X_j|\geq t)\leq 2\exp(-\frac{t^2}{C_j}).
	\eeqrs
	Let $C_{\max} = \max_{1\leq j\leq d}\{C_1, \dots, C_d\}$. Then 
	\beqrs
	\E\max_{1\leq j\leq d} X_j< 2\sqrt{C_{\max} \log d}.
	\eeqrs
\end{lemma}
\begin{proof}
	Let $t^* = \sqrt{C_{\max} \log d}$. 
	We have
	\beqrs
	\E\max_{1\leq j\leq d} X_j&\leq& t^* + d\int_{t^*}^{\infty}\exp(-\frac{t^2}{C_{\max}}) dt\\
	&=& t^* + d\sqrt{\frac{C_{\max}}{2}}\int_{\sqrt{2\log d}}^{\infty}\exp(-\frac{s^2}{2}) ds\\
	&\leq& t^* + \sqrt\frac{C_{\max}}{4 \log d}.
	\eeqrs
	where the first inequality follows by the union bound, the second equality follows by the change of variables: $t = s\sqrt{C_{\max}/2}$, and the last inequality follows by the tail bound $1-\Phi(s) \leq\phi(s)/s$ in Lemma \ref{lem: mill_ratio}. The proof of the lemma is complete by noting $\log d>1/4$ when $d\geq2$. 
\end{proof}

\begin{lemma}\label{bound_W}
	Assume that the canonical kernel $f(U_{kj}, V_{lj})$ defined in (\ref{eq: canonical}) be uniformly bounded by a constant $B$. Recall that
	\beqrs
	W_j = \frac{1}{\sqrt n(n-1)}\sum_{1\leq k\neq l\leq n} f_j(U_{kj}, V_{lj}),
	\eeqrs
	and $W = (W_1, \dots, W_p)$. We have
	\beqrs
	E\max_{1\leq j\leq p} W_j \leq \frac{C\log p}{\sqrt n},
	\eeqrs
	where $C$ is a constant related to $B$.
\end{lemma}

\begin{proof}
	
	By using Theorem 3.5.3 of \cite{de2012decoupling}, i.e.,  the randomization inequality, we have
	\beqrs
	E (\max_{j} \sum_{1\leq k\neq l\leq n} f_j(U_{kj}, V_{lj}))\leq C_1 E(\max_{j} \sum_{1\leq k\neq l\leq n} \eps_k \xi_l f_j(U_{kj}, V_{lj}) )
	\eeqrs
	where $\eps_k$ and $\xi_l$ are independent Rademacher random variables, $k,l=1,\dots, n$. Let $\eps = (\eps_1, \dots, \eps_n)\trans$ and $\xi = (\xi_1, \dots, \xi_n)\trans$. 
	Define 
	\beqr\label{define_akj}
	n^{-1}\sum_{1\leq k\neq l\leq n} \eps_k \xi_l f_j(U_{kj}, V_{lj}) =  \sum_{k=1}^n \eps_k \{n^{-1}\sum_{l=1}^{n}\xi_l f_j(U_{kj}, V_{lj}) \}\defby \sum_{k=1}^n \eps_k a_{k,j}.
	\eeqr
	Thus $a_{k,j}$ is uniformly bounded on $[-B, B]$. Let $a_{j} = (a_{1,j}, \dots, a_{n,j})$. Conditional on $a_j$, by Hoeffding's inequality,
	\beqrs
	\P\left(|\sum_{k=1}^n \eps_k a_{k,j}|\geq t\mid a_j \right) \leq 2\exp\left( -\frac{ t^2}{2\sum_{k=1}^{n}a_{k,j}^2 }\right).
	\eeqrs
	Now we use Lemma \ref{maxE_bound}. By treating $X_j$ as $\sum_{k=1}^n \eps_k a_{k,j}$ and $C_j$ as $2\sum_{k=1}^{n}a_{k,j}^2$, we have
	\beqrs
	\E\{ \max_j \sum_{k=1}^n \eps_k a_{k,j} | (a_1, \dots, a_p)\}< 2\sqrt{\max_j  2\sum_{k=1}^{n}a_{k,j}^2  \log p }.
	\eeqrs
	Taking expectation with respect to $(a_1, \dots, a_p)$, we have
	\beqr\nonumber
	\E\{ \max_j \sum_{k=1}^n \eps_k a_{k,j} \}&<& \sqrt{8\log p}\ \E\sqrt{\max_j  \sum_{k=1}^{n}a_{k,j}^2   }\\\label{bound_akj}
	&\leq&\sqrt{8\log p} \ \sqrt{\E\{\max_j  \sum_{k=1}^{n}a_{k,j}^2   \} },
	\eeqr
	where the last inequality follows by Jensen inequality. Thus it suffices to bound $\E\{\max_j  \sum_{k=1}^{n}a_{k,j}^2\}$. 
	Note that for each $j$, we have
	\beqr\label{decomp_akj}
	\sum_{k=1}^{n}a_{k,j}^2 = n^{-2}\sum_{k=1}^{n}\sum_{l=1}^n	 f_j(U_{kj}, V_{lj})^2 + n^{-2}\sum_{k=1}^{n}\sum_{l_1=1}^n	 \sum_{l_2=1}^n	f_j(U_{kj}, V_{l_1j})f_j(U_{kj}, V_{l_2j})\xi_{l_1}\xi_{1_2}
	\eeqr
	The first term is uniformly bounded by $B^2$. To deal with the second term, let $M^j$ be an $n\times n$ matrix whose $(k,l)$-th element is given by
	\beqrs
	n^{-1}\sum_{k=1}^{n}	f_j(U_{kj}, V_{l_1j})f_j(U_{kj}, V_{l_2j}).
	\eeqrs
	Thus the second term in (\ref{decomp_akj}) is equivalent to $n^{-1}\xi\trans M^j\xi$. By the Hanson-Wright inequality,
	\beqrs
	\P\left(|\xi\trans M^j\xi|>t\mid M^j\right)\leq 2 \exp\left\{-C_1\min\left( \frac{t^2}{\|M^j\|_F^2 }, \frac{t}{\|M^j\|_2} \right) \right\}.
	\eeqrs
	Let $\gamma_1 = \max_j\{\|M^j\|_F\}$, $\gamma_2 = \max_j\{\|M^j\|_2\}$, and $t^* = \max_j\{ \gamma_1\sqrt{\log(p)/C_1}, \gamma_2\log(p)/C_1\}$. By the union bound, 
	\beqrs
	\E\left(\max_j \xi\trans M^j\xi \mid \{M^j\}_{j=1}^d\right) &=& \int_{0}^{\infty}\P\left(\max_j |\xi\trans M^j\xi| \geq t\mid \{M^j\}_{j=1}^d\right)dt \\
	&\leq& t^* + 2p \int_{t^*}^{\infty}\max\left\{\exp(-\frac{C_1 t^2}{\gamma_1^2}) , \exp(-\frac{C_1 t}{\gamma_2})\right\}dt.
	\eeqrs
	By Lemma \ref{lem: mill_ratio}, $1-\Phi(s) \leq\phi(s)/s$. Thus
	\beqrs
	2p\int_{t^*}^{\infty}\exp(-\frac{C_1 t^2}{\gamma_1^2}) dt \leq \frac{\gamma_1}{\sqrt{C_1\log p }}.
	\eeqrs
	Moreover,
	\beqrs
	2p\int_{t^*}^{\infty}\exp(-\frac{C_1 t}{\gamma_2}) dt \leq \frac{2\gamma_2}{C_1}.
	\eeqrs
	Because $\gamma_2<\gamma_1$ and each element in $M^j$ being bounded by $B^2$ leads to $\gamma_1\leq n B^2$, we obtain that
	\beqrs
	\E\left(\max_j n^{-1}\xi\trans M^j\xi \mid \{M^j\}_{j=1}^p\right)\leq C_2 \frac{(\log p) \gamma_1}{n}\leq C_2B^2\log p.
	\eeqrs
	Following (\ref{bound_akj}) and (\ref{decomp_akj}), we have $\E\{ \max_j \sum_{k=1}^n \eps_k a_{k,j} \}< C_3\log p$.	Thus we have 
	\beqrs
	E\max_{1\leq j\leq p} W_j \leq \frac{C\log p}{\sqrt n},
	\eeqrs
	where $C$ is a constant related to $B$. The proof is completed.
\end{proof}

\subsection{Proof of Theorem \ref{thm: ss_consistency}}

Denote $p\defby M-1$. 
We first quantify the difference between $\wh\Gamma$ and $\Gamma$. 	By definition, $\wh \Gamma_g $ is clearly a consistent estimator of $\Gamma_g $. Both $\Gamma_g $ and $\wh\Gamma_g $ are $p\times p$ matrix with bounded elements. By the bounded difference inequality, for each $(j, j')$, we have
\beqrs
\P(|\wh\Gamma_{g, (j, j')} - \Gamma_{g, (j, j')}| \geq  t)  \leq 2\exp(-C n t^2).
\eeqrs
Let $N = \max\{n,p\}$.
By the union bound, we have
\beqr\label{eq: gamma_bound}
\P(\max_{j, j'}|\wh\Gamma_{g, (j, j')} - \Gamma_{g, (j, j')}| \geq C\sqrt{\log(N)/n})  \leq N^{-1}.
\eeqr
Let $E = \{\max_{j, j'}|\wh\Gamma_{g, (j, j')} - \Gamma_{g, (j, j')}| \leq C\sqrt{\log(N)/n}\}$. Then under the event $E$,  by Theorem 2 of \cite{chernozhukov2015comparison}, we have
\beqr\label{ineq: z}
z(\alpha, \wh\Gamma)\leq z(\alpha +C_1 (\log^{5/6}N) n^{-1/6}, \Gamma).
\eeqr

For simplicity, we consider $m$ as fixed and drop $m$ in the notation. And we let the number of bootstraps $B$ go to infinity and ignore the bootstrap variability. Recall that $\wh p$ is the bootstrap $p$-value obtained by Algorithm 1. It follows that
\beqrs
&&\P(\wh p\leq\alpha) =\P(\min_{j\neq m}\wh\mu_{m,j} \leq z(\alpha, \wh\Gamma))\\ 
\leq &&\P(\min_{j\neq m}\wh\mu_{m,j} \leq z(\alpha, \wh\Gamma), E) + \P(E^c)\\ 
\leq &&\P(\min_{j\neq m}\wh\mu_{m,j} \leq z(\alpha +C_1 (\log^{5/6}N) n^{-1/6}, \Gamma), E ) + N^{-1}\\ 
\leq&&\P(\min Z_{\Gamma} \leq z(\alpha +C_1 (\log^{5/6}N) n^{-1/6}, \Gamma) ) +  C  n^{-K_1}(\log p)^7+ N^{-1}\\
=&&\alpha +C_1 (\log^{5/6}N) n^{-1/6}+  C  n^{-K_1}(\log p)^7+ N^{-1}.
\eeqrs
The first inequality holds by basic probability algebra. The second inequality holds by (\ref{eq: gamma_bound}) and (\ref{ineq: z}). And the last inequality holds because of Theorem \ref{thm: gauss_approxi}. This finishes the proof.

\subsection{Proof of Theorem \ref{thm: screening_consistency}}

Based on the result of Theorem \ref{thm: ss_consistency}, it suffices to show that there exists a boundary set near the null hypothesis that is contained in $\wh J_m$ with high probability. Let 
\beqrs
J_m = \{j\neq m: \frac{\sqrt n\mu_{m,j}}{\sigma_{m,j}}< c_{\alpha'} \}.
\eeqrs
The proof consists of three steps. 

{\bf Step 1:} We show that 
\beqrs
\P(\wh\mu_{m,j}\geq 0, \forall j \in J_m^c) \rightarrow 1.
\eeqrs
Note that
\beqrs
\P(\exists j\in J_m^c, \wh\mu_{m,j}< 0) &\leq& \P(\min_{1\leq j\leq p}\frac{\sqrt n(\wh\mu_j - \mu_j)}{\sigma_j} <- c_{\alpha'} ) \\
&\leq&   \P\Big(\min_{1\leq j\leq p}\frac{\sqrt n(\wh\mu_j - \mu_j)}{\wh\sigma_j} <- (1-r) c_{\alpha'} \Big)  + \P\Big(\max_{1\leq j\leq p}|\frac{\sigma_j}{\wh\sigma_j} - 1|>r \Big)\\
&\leq& \sum_{j=1}^{p}\P\Big(\frac{\sqrt n(\wh\mu_j - \mu_j)}{\wh\sigma_j} <- (1-r)c_{\alpha'}\Big) + o(1)\\
&=&  o(1) ,
\eeqrs
where the first inequality follows the definition of $J_m^c$. The second inequality follows by noting that $\sigma_j\geq \wh\sigma_j(1-r)$ when  $\sigma_j/\wh\sigma_j - 1\geq -r$ for some $0<r<1$. The third inequality follows by the union bound and Lemma \ref{lem: sigma_bound}. To deal with the last equality, note that by using the bounded difference inequality, we get $\P(\wh\mu_j - \mu_j\geq t)\leq \exp(-nt^2/8)$. Thus it follows that $\P(\sqrt{n}(\wh\mu_j - \mu_j)\geq t')\leq \exp(-t'^2/8)$.  Furthermore,  by Lemma \ref{lem: mill_ratio},
\beqrs
c_{\alpha'} = \Phi^{-1}\Big(1-\frac{\alpha'}{p^{1+s}}\Big)\asymp \sqrt{\log p^{1+s}}.
\eeqrs
Thus
\beqrs
\P\Big(\frac{\sqrt n(\wh\mu_j - \mu_j)}{\wh\sigma_j} <- (1-r)c_{\alpha'}\Big)\leq\exp(-C c_{\alpha'}^2) =\frac{C}{p^{1+s}}.
\eeqrs
By summing over  $j=1,\dots,p$, we get the bound is of order $C/p^r$, which goes to zero as $p\rightarrow\infty$.

{\bf Step 2:} We show that
\beqrs
\P(J_m\subset \wh J_m) \rightarrow 1.
\eeqrs
Note that 
\beqr\label{eq: step2}
J_m\not \subset \wh J_m  \Rightarrow \exists j, \mbox{such that} \frac{\sqrt n\mu_{m,j}}{\sigma_{m,j}}<  c_{\alpha'}, \frac{\sqrt n\wh\mu_{m,j}}{\wh\sigma_{m,j}} > 2 c_{\alpha'}.
\eeqr
Thus we have
\beqrs
\P( J_m\not \subset \wh J_m) &\leq& \P\Big(\min_{1\leq j\leq p} \sqrt n(\mu_j -\wh\mu_j) + (2\wh\sigma_j - \sigma_j ) c_{\alpha'} <0  \Big)\\
&\leq & \P\Big( \min_{1\leq j\leq p} \sqrt n(\mu_j -\wh\mu_j) +\wh\sigma_j(1-r) c_{\alpha'} <0   \Big)+ \P\Big(\max_{1\leq j\leq p}|\frac{\sigma_j}{\wh\sigma_j} - 1|>r \Big)\\
&=&\P\Big( \min_{1\leq j\leq p} \frac{\sqrt n(\mu_j -\wh\mu_j)}{\wh\sigma_j} < -(1-r) c_{\alpha'}   \Big)+ \P\Big(\max_{1\leq j\leq p}|\frac{\sigma_j}{\wh\sigma_j} - 1|>r \Big).
\eeqrs
The first inequality holds by (\ref{eq: step2}). The second inequality holds because $\wh\sigma_j(1-r)\leq 2\wh\sigma_j - \sigma_j $ when $\sigma_j/\wh\sigma_j - 1 \leq r$.
Then the result of step 2 follows similar reasoning in step 1.

{\bf Step 3:} We are ready to prove the main argument of the theorem. When $J_m = \emptyset$, we know from Step 1 that $\wh\mu_{m,j}\geq0$, $\forall j \in J_m^c$ with probability larger than $1-\alpha'+o(1)$, which implies that $H_{0,m}$ holds with probability larger than $1-\alpha'+o(1)$. Thus we have $\P(m\in\calA_{ss}) \geq 1 - \alpha' + o(1)$.

When $|J_m|\geq 1$, conditional on the event of $\{\wh\mu_{m,j}\geq0$, $\forall j\in J_m^c\}\cap \{J_m\subset \wh J_m\}$,  we have
\beqrs
\P(\wh p\leq\alpha) &=&\P(\min_{j\in \wh J_m}\wh\mu_{m,j} \leq z(\alpha, \wh\Gamma))  =\P(\min_{j\in J_m}\wh\mu_{m,j} \leq z(\alpha, \wh\Gamma)) \\
&\leq&\P(\min_{j\in J_m}\wh\mu_{m,j} \leq z(\alpha, \wh\Gamma_{J_m}))\leq \alpha + o(1).
\eeqrs
The first equality follows by definition of $p$-value, and the second equality follows the conditional event. The first inequality follows by the fact that $z(\alpha, \wh\Gamma)\leq z(\alpha, \wh\Gamma_{J_m})$, and the last inequality follows by Theorem \ref{thm: ss_consistency}. Thus we have
\beqrs
\P(m\in\calA_{ss})\geq 1 - \alpha +o(1).
\eeqrs

\subsection{Proof of Theorem \ref{thm: gauss_approxi}}

{\bf Preliminaries.}
In this section, we first define the smooth approximation of the max function and indicator function.

Let $Z = (Z_1,\dots, Z_d)\in\mR^p$. For any $\beta>0$, define the smooth max function $F_\beta(Z)$ 
\beqrs
F_\beta(Z) = \beta^{-1}\log\{\sum_{j=1}^{d}\exp(\beta Z_j) \}. 
\eeqrs
$\beta$ is the smoothing parameter that controls the level of approximation. Simple calculations show that
\beqrs
0\leq F_\beta(Z) - \max_{1\leq j\leq d}Z_j \leq e_\beta,
\eeqrs
where $e_\beta = \beta^{-1}\log p$.	Define $\pi_j(Z) = \partial{F_\beta(Z)}/\partial Z_j$. By Lemma A.2 of \cite{chernozhukov2013gaussian}, we have $\pi_j(Z)\geq0$ and $\sum_{j=1}^{d}\pi_j(Z) = 1$. 

Now we define the smoothed indicator function $g(\cdot)$. Let $\C^k$ denote the class of $k$ times continuously differentiable functions that are bounded, and $g_0(\cdot) :\mR\rightarrow [0, 1] \in\C^3$. Let $g_0(s) = 1$ if $s\leq0$ and $g_0(s) = 0$ if $s\geq1$. Further assume that $\sup_{s\in\mR}|\partial^qg_0(s)|\leq C_0$ for some constant $C_0>0$ and $q = 0, 1, 2, 3$.
Define $g(s) = g_0(\psi (s - t- e_\beta))$ for some $\psi>0$. It follows that
\beqrs
&&\quad \sup_s|\partial ^q g(s)| \leq C_0 \psi^q, \\
&& \indic(s\leq t + e_\beta)\leq g(s)\leq \indic(s\leq t + e_\beta + \frac{1}{\psi}).
\eeqrs
Clearly, a large $\beta$ and $\psi$ will lead to better approximations of the corresponding non-smooth functions. 

\begin{proof}

	By the triangle inequality,
	\beqrs
	\rho(T, Z) \leq  \sup_{t\in\mR}\big|\P(\|T\|\leq t)  - \P(\|L\|\leq t) \big|+\sup_{t\in\mR}\big|\P(\|L\|\leq t)  - \P(\|Z\|\leq t) \big|.
	\eeqrs 
	Note that the second term can be bounded by Lemma 2.3 in \cite{chernozhukov2013gaussian}. 
	To bound the first term, we begin with bounding $|E\{ \lambda (T) - \lambda(L)\}|$, where $\lambda = g\circ F_\beta$.
	Through Taylor's expansion, 
	\beqrs
	|E\{ \lambda (T) - \lambda(L)\}| = |E\{\sum_j W_j \partial g(F_\beta(\xi))\pi_j(\xi)\}|\leq C_0\psi E\max_{1\leq j\leq p} W_j,
	\eeqrs
	where $\xi$ is a random vector on the line segment between $T$ and $L$. The inequality holds because  $\pi_j(Z)\geq0$ and $\sum_{j=1}^{d}\pi_j(Z) = 1$. 
	
	By the properties of $F_{\beta}(\cdot)$ and $g(\cdot)$,  it follows that
	\beqr\nonumber
	&&\P(\|T\|\leq t)  \leq \P(F_\beta(T)\leq t + e_\beta) \leq \E\{g(F_\beta(T)) \} = \E\{\lambda(T)\}\\\nonumber
	&&\leq \E\{\lambda(L)\} + C_0\psi E\|W\| _{\infty}= \E\{g(F_\beta(L)) \} + C_0\psi E\max_{1\leq j\leq p} W_j\\\nonumber
	&&\leq \P(F_\beta(L)\leq t+e_\beta+\psi^{-1}) + C_0\psi E\max_{1\leq j\leq p} W_j\\\label{dif_TL}
	&&\leq  \P(\|L\|\leq t+e_\beta+\psi^{-1}) + C_0\psi E\max_{1\leq j\leq p} W_j. 
	\eeqr
	By Lemma 2.1 of \cite{chernozhukov2013gaussian}, i.e., the anti-concentration inequality,	
	\beqr\label{diff_LL}
	\P(\|L\|\leq t+e_\beta+\psi^{-1}) - \P(\|L\|\leq t)\leq C_1 (e_\beta+\psi^{-1})\sqrt{1\vee \log(d\psi)}.
	\eeqr
	By equation (\ref{dif_TL}) and  (\ref{diff_LL}), together with Lemma \ref{bound_W}, we have
	\beqrs
	\P(\|T\|\leq t)  - \P(\|L\|\leq t)\leq C_2 \frac{(\log p) \psi}{\sqrt n} + C_1 (e_\beta+\psi^{-1})\sqrt{1\vee \log(d\psi)}.
	\eeqrs
	To obtain a clear bound, we choose $\beta$ such that $e_\beta$ and $\psi$ are balanced. 
	
	By Lemma 2.3 of \cite{chernozhukov2013gaussian},  
	\beqrs
	\sup_{t\in\mR}\big|\P(\|L\|\leq t)  - \P(\|Z\|\leq t) \big|\leq C_3 n^{-1/8}(\log p)^7.
	\eeqrs
	Thus we have
	\beqrs
	\P(\|T\|\leq t)  - \P(\|Z\|\leq t)\leq C_1 (e_\beta+\psi^{-1})\sqrt{1\vee \log(d\psi)}+C_2 \frac{(\log p) \psi}{\sqrt n}  + C_3 n^{-1/8}(\log p)^7.
	\eeqrs
	The other half of the inequality follows similarly. Thus we have
	\beqrs
	\big|\P(\|T\|\leq t)  - \P(\|Z\|\leq t)\big|\leq C_1 (e_\beta+\psi^{-1})\sqrt{1\vee \log(d\psi)}+C_2 \frac{(\log p) \psi}{\sqrt n}  + C_3 n^{-1/8}(\log p)^7.
	\eeqrs	
	Let $\beta = n^{s_1}$, $\psi = n^{s_2}$, we obtain that under the high dimensional setting
	\beqrs
	\sup_{t\in\mR}\big|\P(\|T\|\leq t)  - \P(\|Z\|\leq t)\big|\leq C  n^{-K_1}(\log p)^7,
	\eeqrs	
	where $K_1 = \min\{s_1, s_2, 1/2 - s_2, 1/8\}$.
	This completes the proof.
	
\end{proof}

\section{Additional Simulations}

\begin{figure}
	\centering
	\begin{subfigure}[b]{0.45\textwidth}
		\centering
		\includegraphics[scale=0.28]{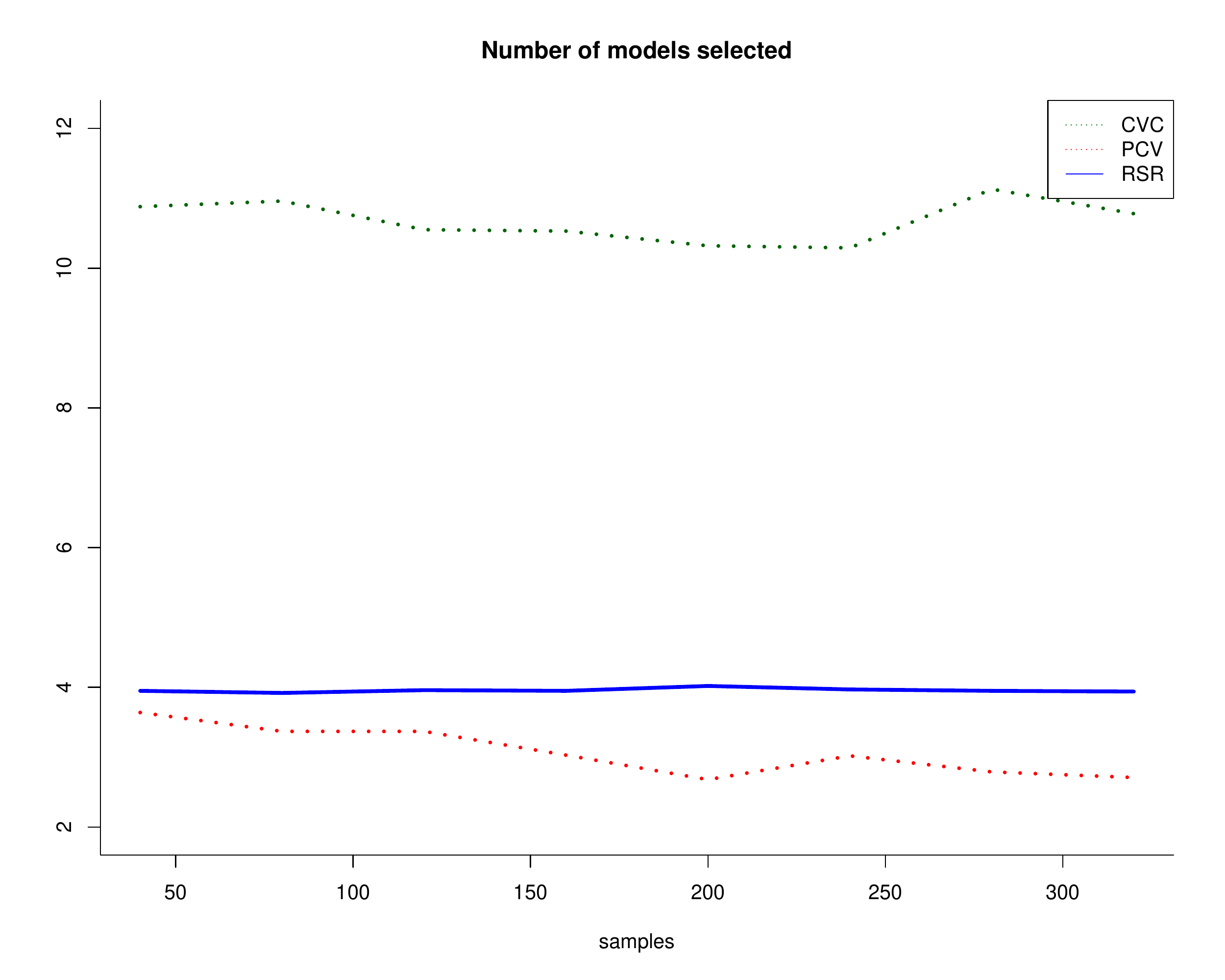}
	\end{subfigure}
	\hfill
	\begin{subfigure}[b]{0.45\textwidth}
		\centering
		\includegraphics[scale=0.28]{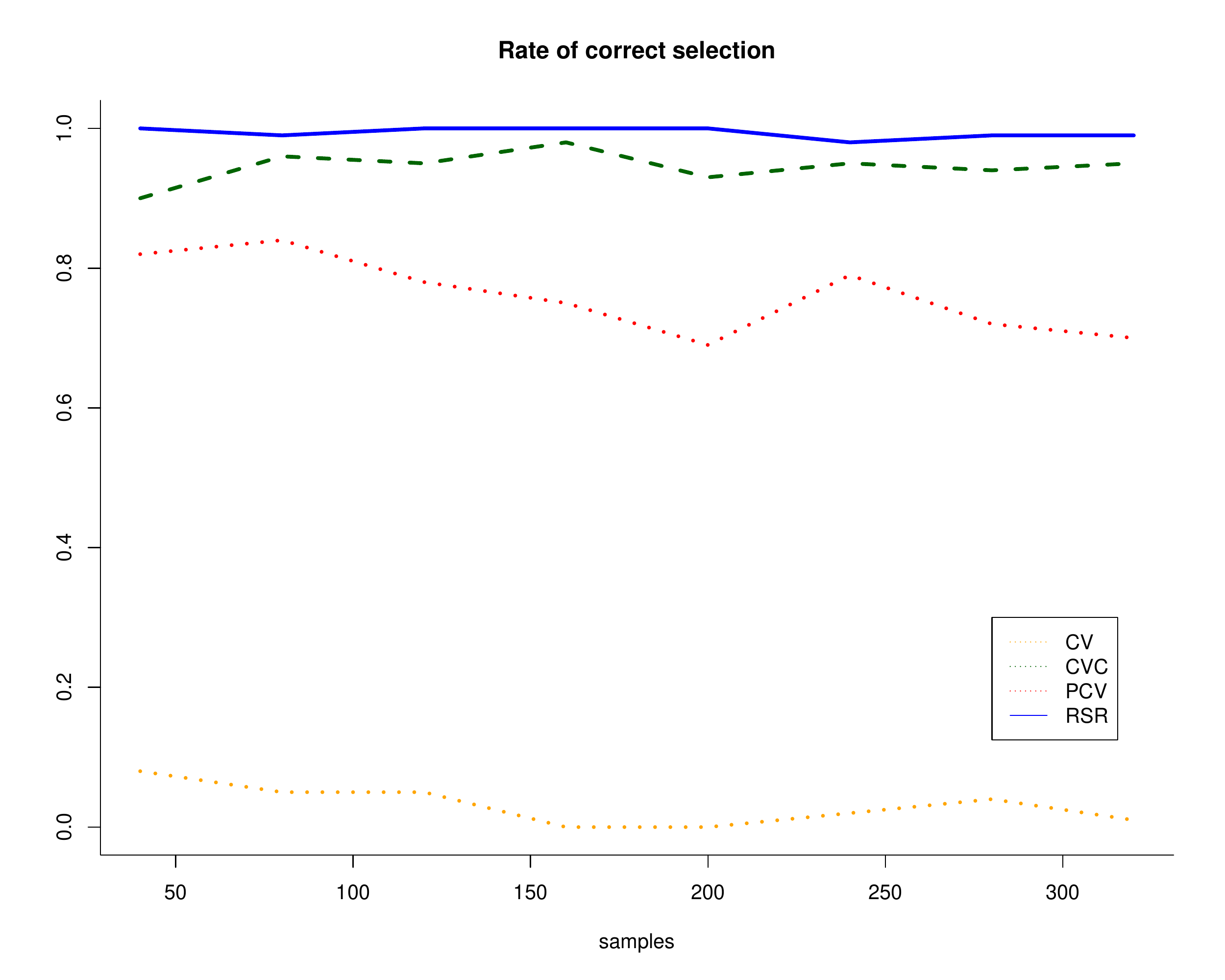}
	\end{subfigure}
	\begin{subfigure}[b]{0.45\textwidth}
		\centering
		\includegraphics[scale=0.3]{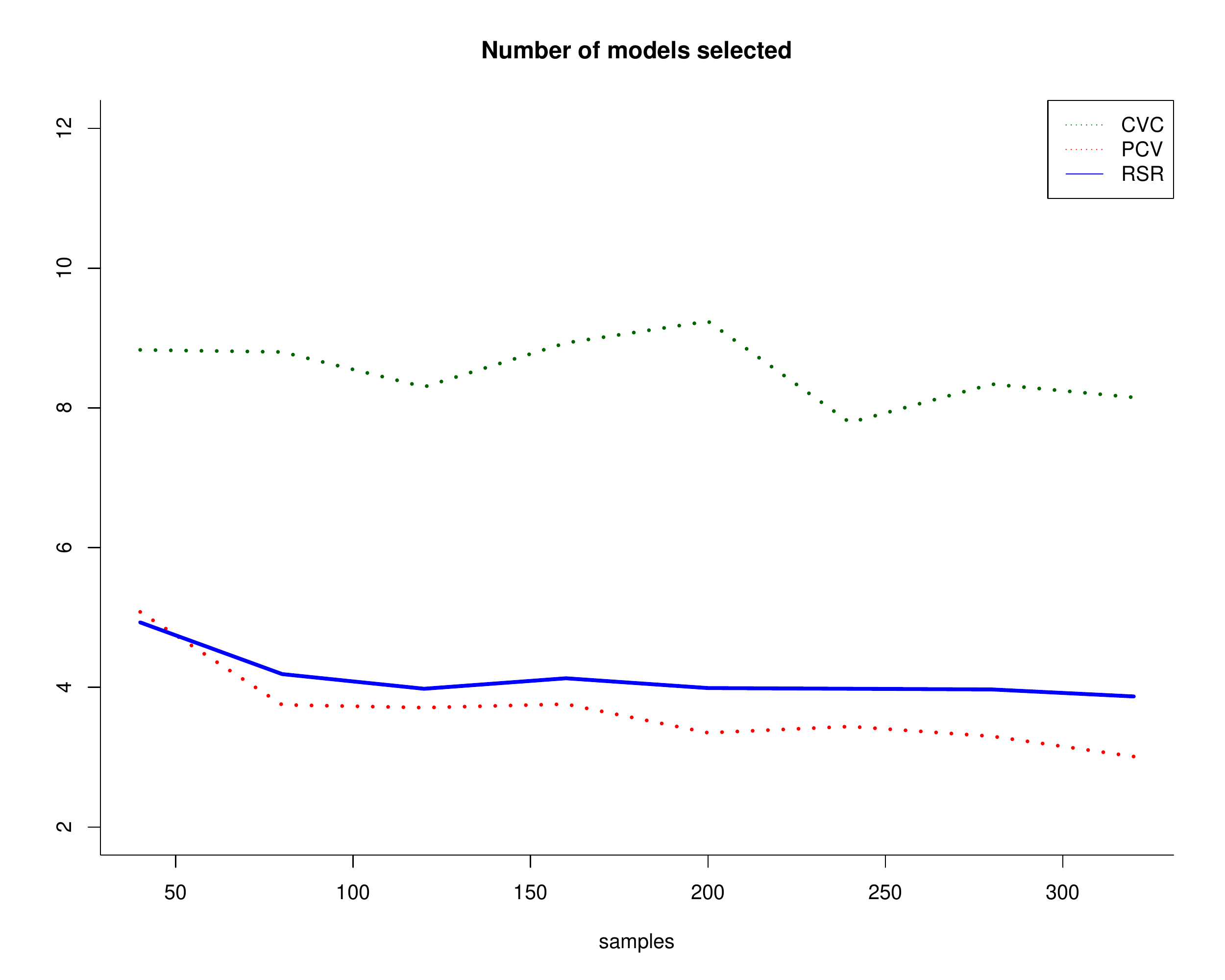}
	\end{subfigure}
	\hfill
	\begin{subfigure}[b]{0.45\textwidth}
		\centering
		\includegraphics[scale=0.3]{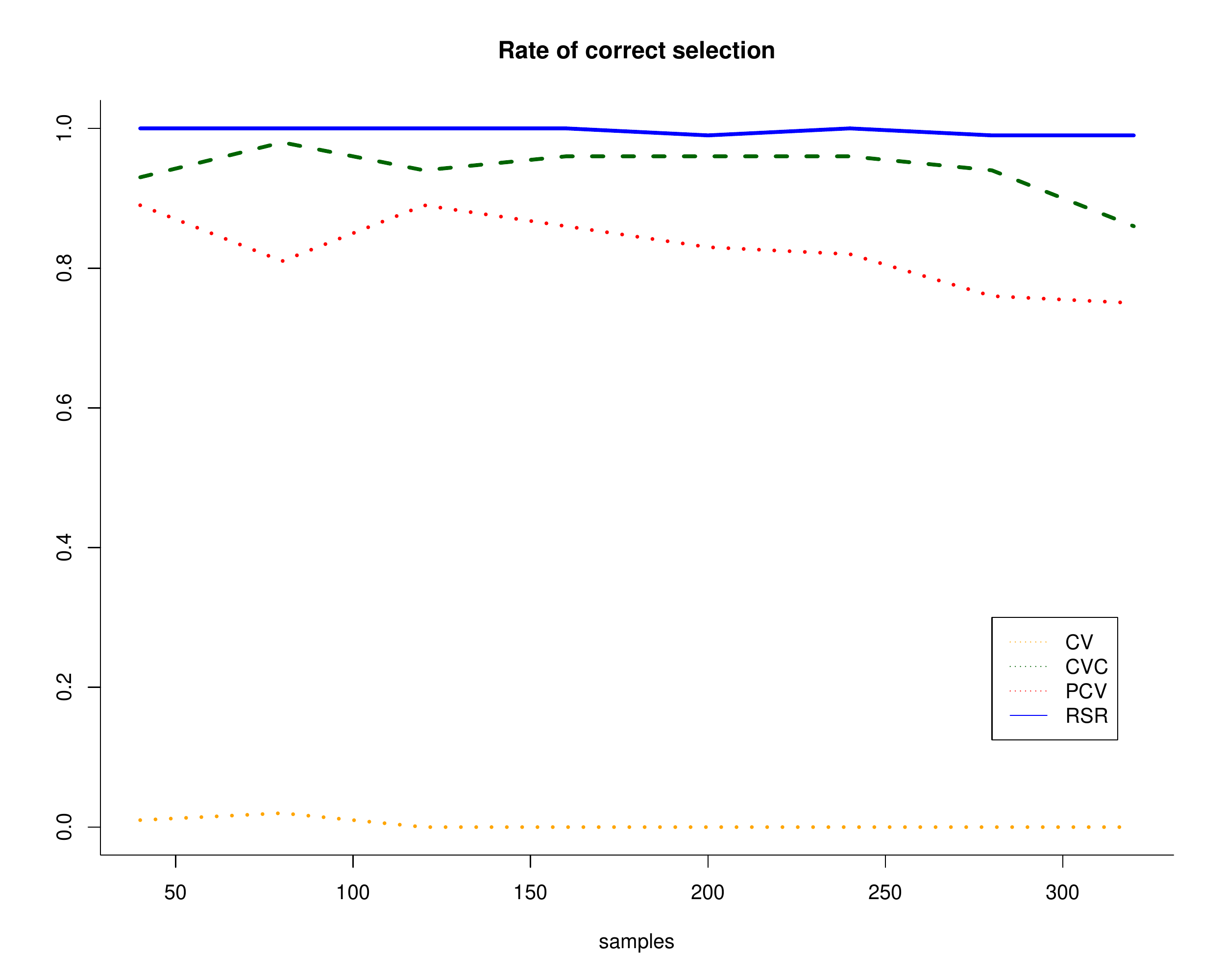}
	\end{subfigure}
	\begin{subfigure}[b]{0.45\textwidth}
		\centering
		\includegraphics[scale=0.3]{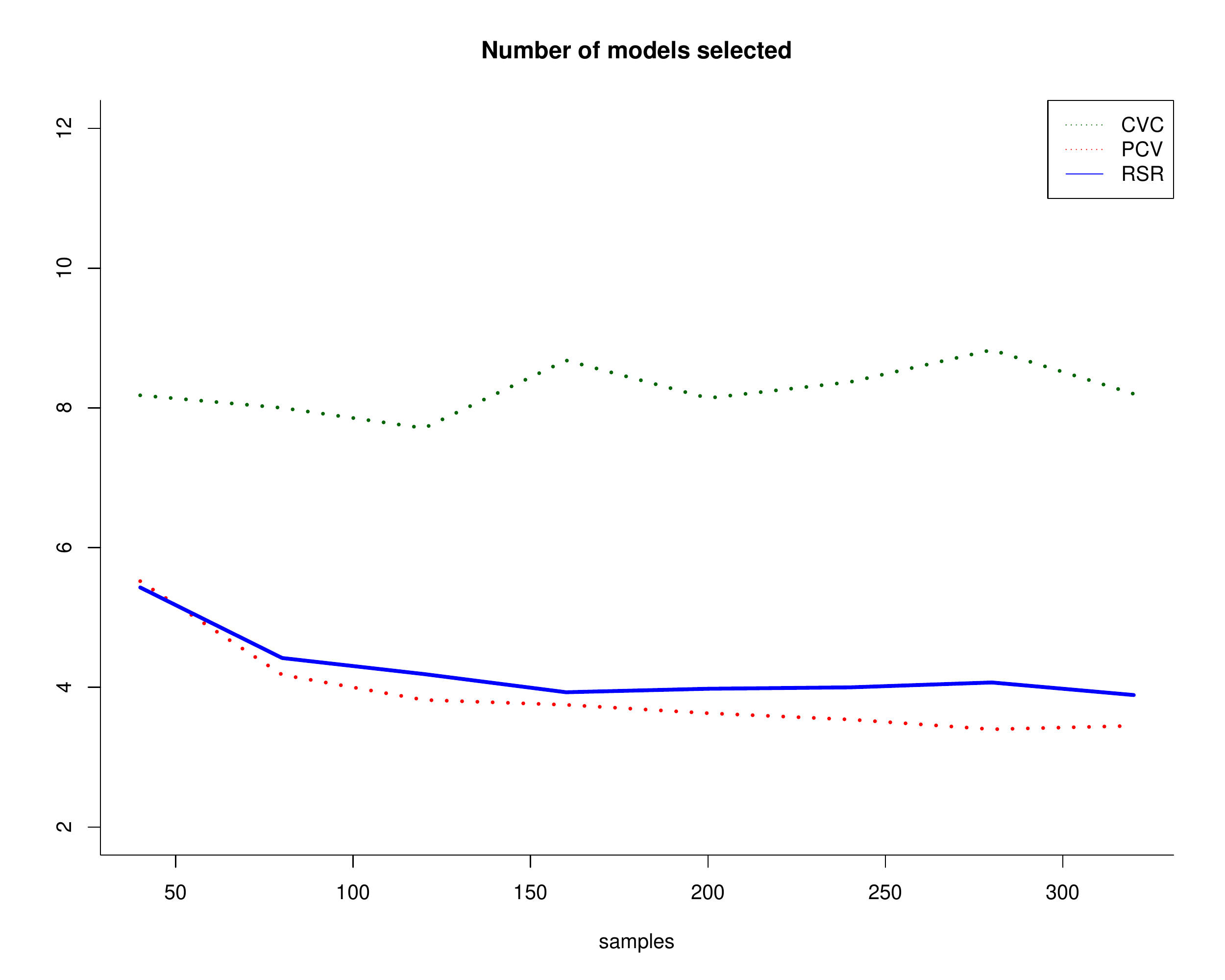}
	\end{subfigure}
	\hfill
	\begin{subfigure}[b]{0.45\textwidth}
		\centering
		\includegraphics[scale=0.3]{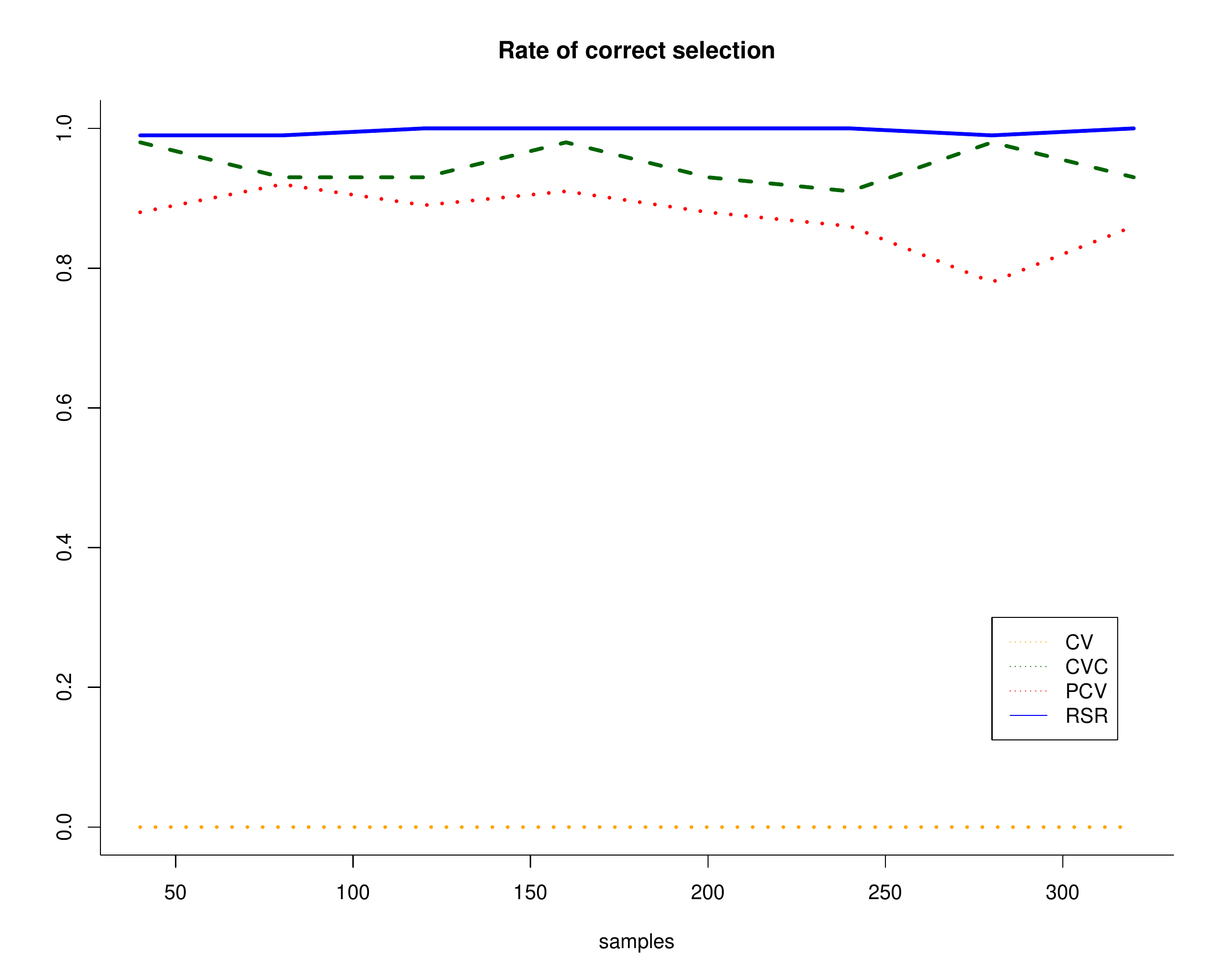}
	\end{subfigure}
	\caption{ The left column: average selected model size. The right column: the rate of correct selection. From the top row to the third row, $X$ is generated from $t_1$, $t_2$, and $t_3$, respectively. $\alpha = 0.05$   }\label{subset2}
\end{figure}

\begin{figure}
	\centering
	\begin{subfigure}[b]{0.45\textwidth}
		\centering
		\includegraphics[scale=0.35]{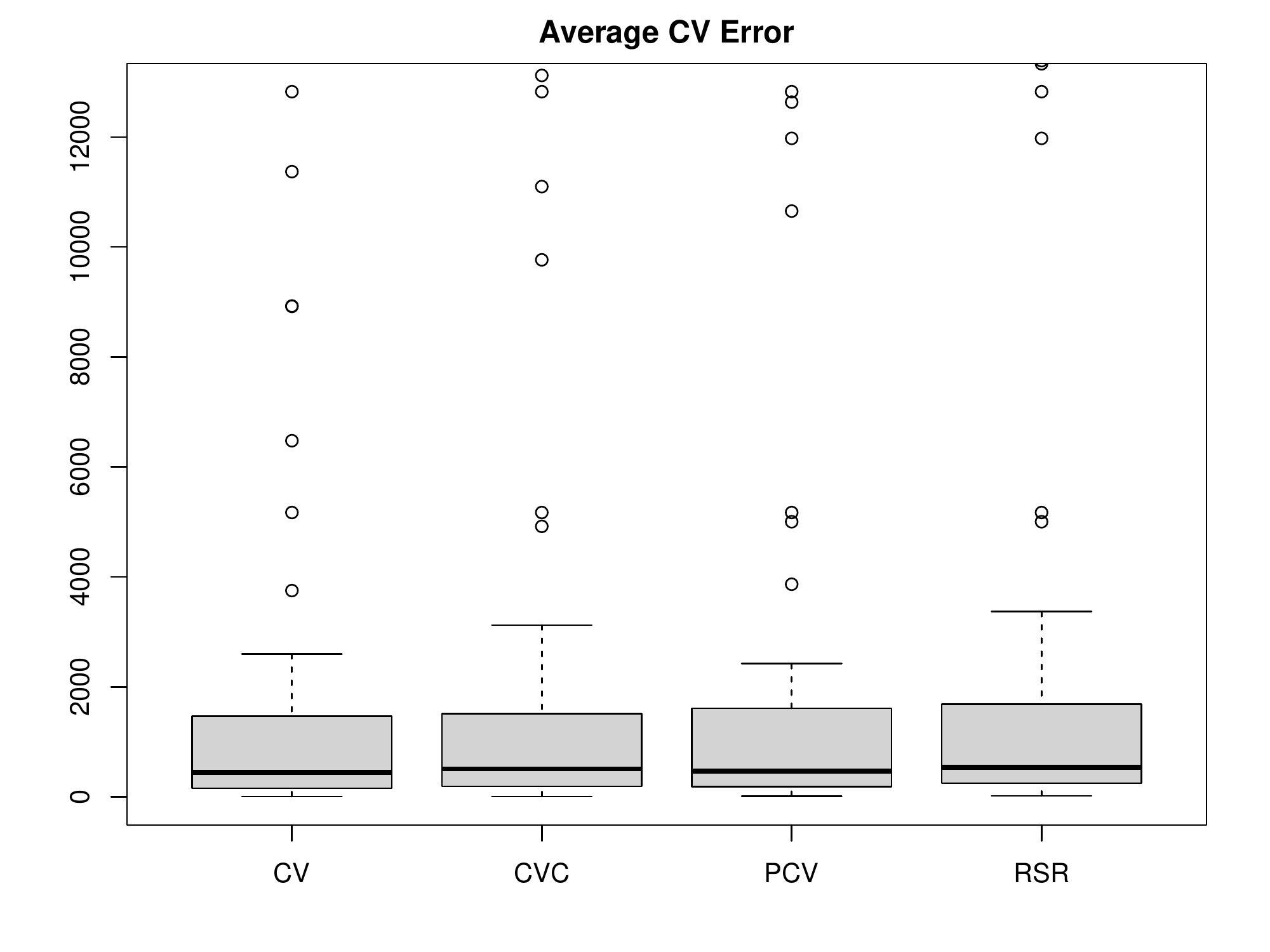}
	\end{subfigure}
	\hfill
	\begin{subfigure}[b]{0.45\textwidth}
		\centering
		\includegraphics[scale=0.35]{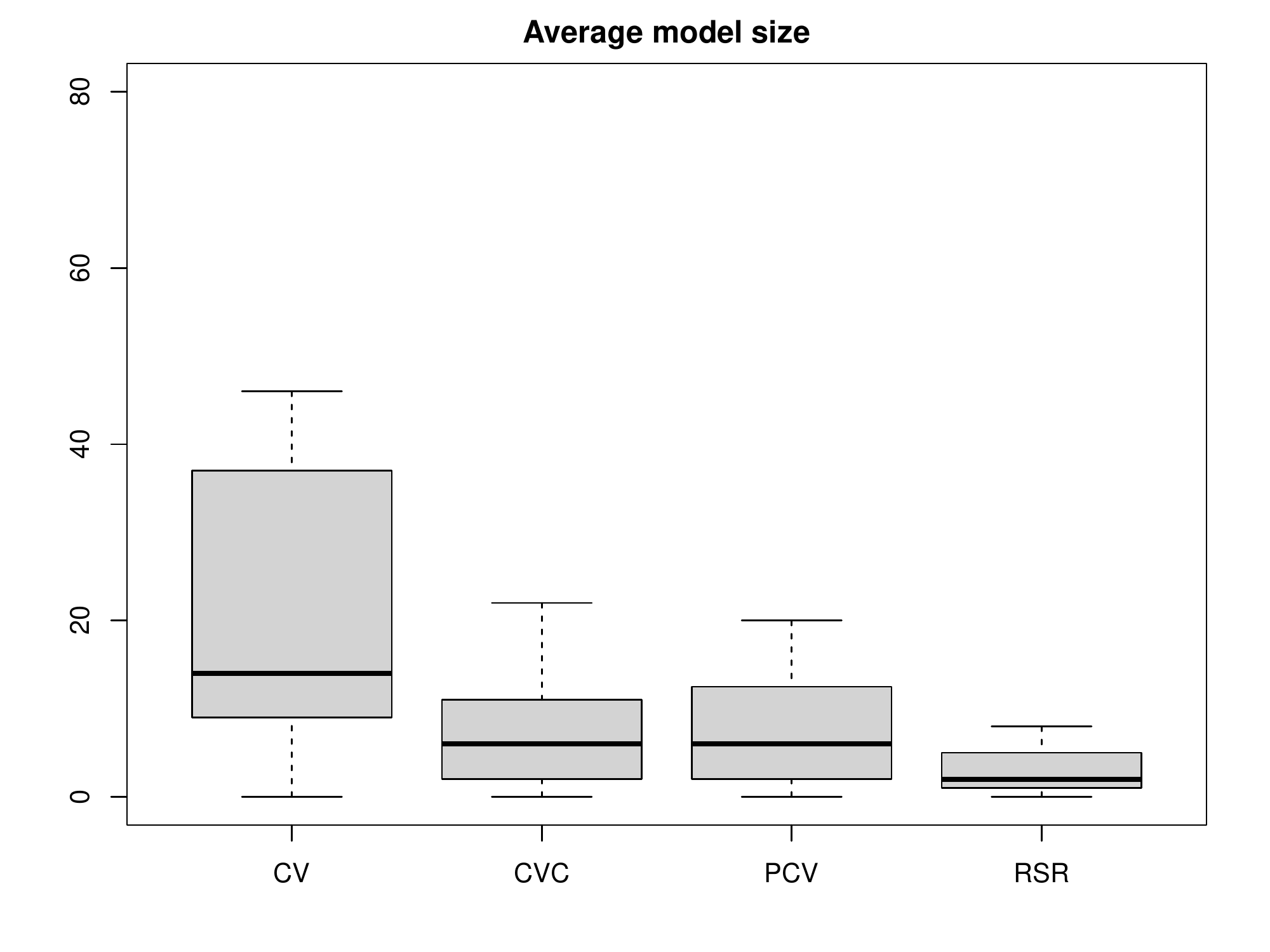}
	\end{subfigure}
	\begin{subfigure}[b]{0.45\textwidth}
		\centering
		\includegraphics[scale=0.35]{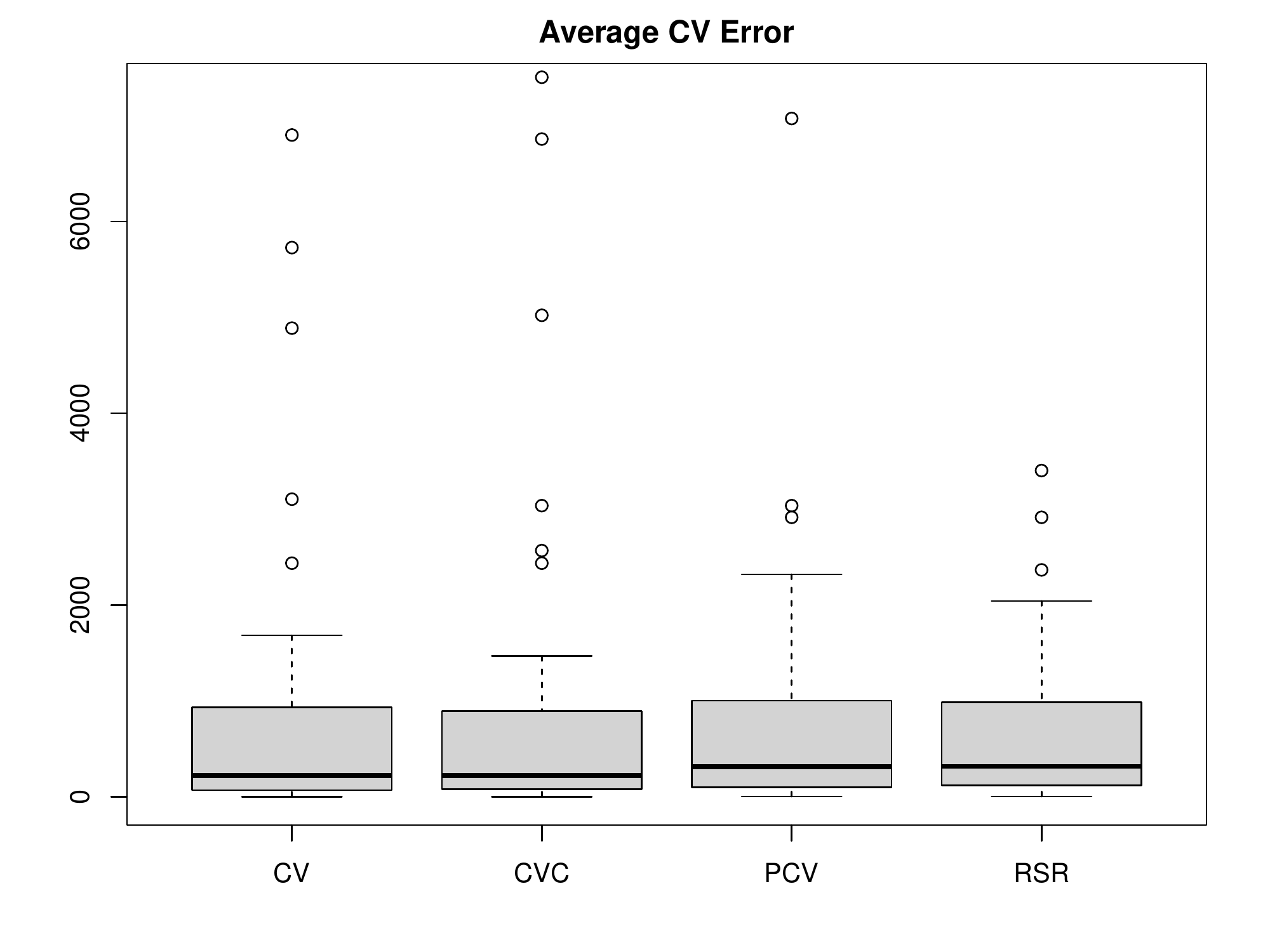}
	\end{subfigure}
	\hfill
	\begin{subfigure}[b]{0.45\textwidth}
		\centering
		\includegraphics[scale=0.35]{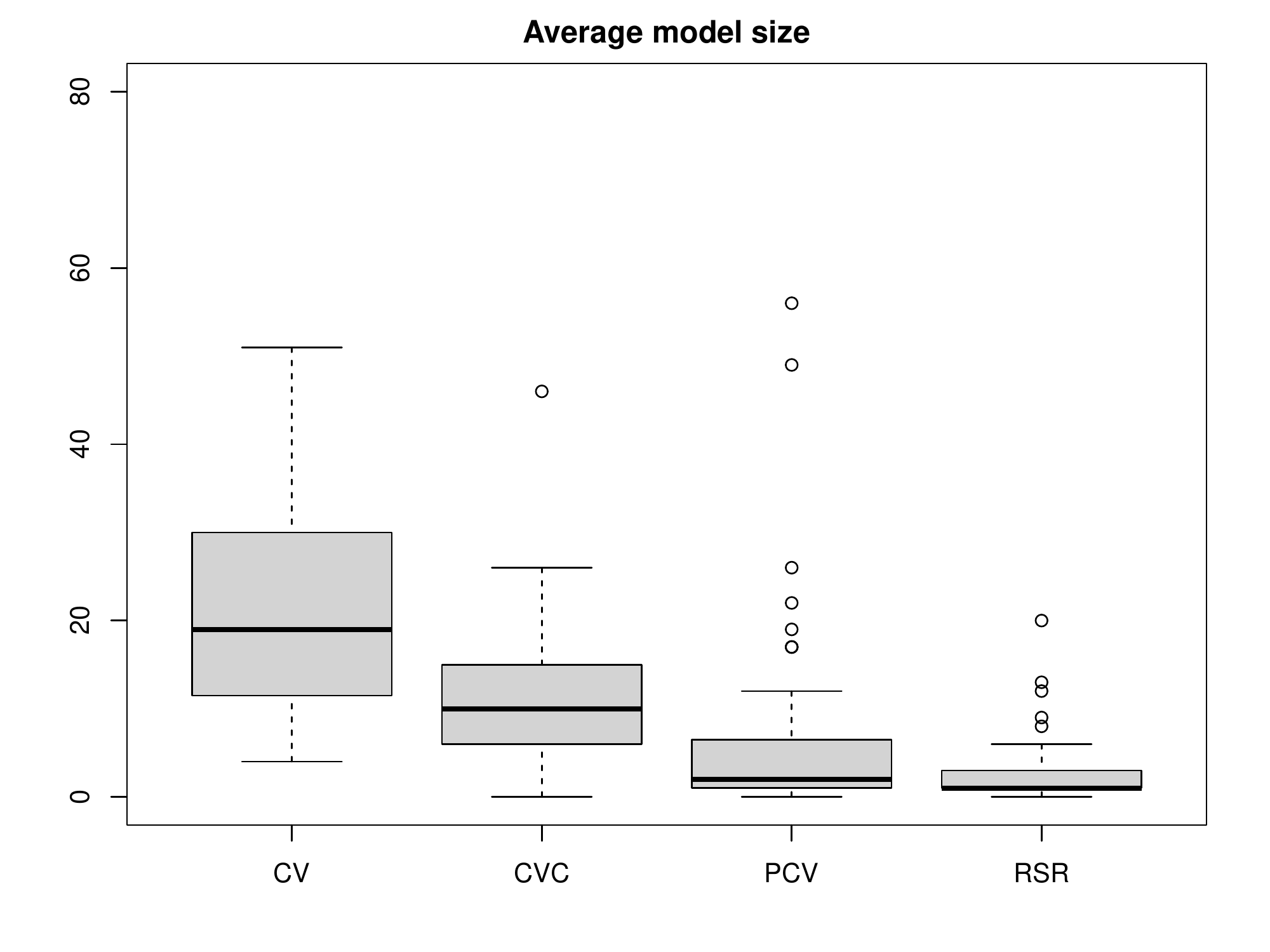}
	\end{subfigure}
	\caption{ Test error and selected model size. The first row represents Lasso without truncation on the data. The second row represents Lasso with truncation on the 99\% quantile of both $X$ and $Y$. $\alpha = 0.01$.}\label{fig: real2}
\end{figure}

 \end{document}